\documentclass[a4paper,superscriptaddress,english]{quantumarticle}
\pdfoutput=1
\usepackage[T1]{fontenc}
\usepackage[utf8]{inputenc}
\usepackage{lipsum,babel}

\usepackage[matrix,frame,arrow]{xy}
\usepackage{amsmath}
\newcommand{\bra}[1]{\left\langle{#1}\right\vert}
\newcommand{\ket}[1]{\left\vert{#1}\right\rangle}
\newcommand{\qw}[1][-1]{\ar @{-} [0,#1]}
\newcommand{\multigate}[2]{*+<1em,.9em>{\hphantom{#2}} \qw \POS[0,0].[#1,0];p !C *{#2},p \save+LU;+RU **\dir{-}\restore\save+RU;+RD **\dir{-}\restore\save+RD;+LD **\dir{-}\restore\save+LD;+LU **\dir{-}\restore}

\newcommand{\ghost}[1]{*+<1em,.9em>{\hphantom{#1}} \qw}

\newcommand{\ustick}[1]{*!D!<0em,-.2em>=<0em>{{\scriptstyle #1}}}

\newcommand{\Qcircuit}[1][0em]{\xymatrix @*=<#1>} 
\newcommand{\node}[2][]{{\begin{array}{c} \ _{#1}\  \\ {#2} \\ \ \end{array}}\drop\frm{o} }

\usepackage{mathtools}
\usepackage{amsthm}
\usepackage{enumitem}
\usepackage{amsmath,amssymb,mathrsfs}
\usepackage{hyperref}
\usepackage{graphicx}
\usepackage{color}

\usepackage{tikzit}
\usepackage{wasysym}
\usepackage{amsfonts}
\usepackage{stmaryrd}
\usepackage{verbatim}
\usepackage[toc,page]{appendix}

\tikzstyle{texthere}=[inline text, font={\footnotesize}]
\tikzstyle{tiny box}=[rectangle, inline text, fill=white, draw, minimum height=5mm, yshift=-0.5mm, minimum width=5mm, font={\small}]
\tikzstyle{small box}=[rectangle, inline text, fill=white, draw, minimum height=7.5mm, yshift=-0.5mm, minimum width=7.5mm, font={\small}]
\tikzstyle{medium box}=[rectangle, inline text, fill=white, draw, minimum height=10mm, yshift=-0.5mm, minimum width=7.5mm, font={\small}]
\tikzstyle{semilarge box}=[rectangle, inline text, fill=white, draw, minimum height=12.5mm, yshift=-0.5mm, minimum width=7.5mm, font={\small}]
\tikzstyle{almost large box}=[rectangle, inline text, fill=white, draw, minimum height=15mm, yshift=-0.5mm, minimum width=7.5mm, font={\small}]
\tikzstyle{big box}=[rectangle, inline text, fill=white, draw, minimum height=17.5mm, yshift=-0.5mm, minimum width=7.5mm, font={\small}]
\tikzstyle{large box}=[rectangle, inline text, fill=white, draw, minimum height=20mm, yshift=-0.5mm, minimum width=7.5mm, font={\small}]
\tikzstyle{very large box}=[rectangle, inline text, fill=white, draw, minimum height=22.5mm, yshift=-0.5mm, minimum width=7.5mm, font={\small}]
\tikzstyle{wide very large box}=[rectangle, inline text, fill=white, draw, minimum height=20mm, yshift=-0.5mm, minimum width=10mm, font={\small}]
\tikzstyle{upground}=[circuit ee IEC, thick, ground, scale=2.1]
\tikzstyle{downground}=[circuit ee IEC, thick, ground, rotate=180, scale=2.1]
\tikzstyle{point}=[regular polygon, regular polygon sides=3, draw, scale=0.67, inner sep=-0.5pt, minimum width=9mm, fill=white, regular polygon rotate=90]
\tikzstyle{copoint}=[regular polygon, regular polygon sides=3, draw, scale=0.67, inner sep=-0.5pt, minimum width=9mm, rotate=-90, fill=white]
\tikzstyle{wide copoint}=[fill=white, draw, shape=isosceles triangle, shape border rotate=180, isosceles triangle stretches=true, inner sep=0pt, minimum width=1.5cm, minimum height=6.12mm]
\tikzstyle{wide point}=[fill=white, draw, shape=isosceles triangle, shape border rotate=0, isosceles triangle stretches=true, inner sep=0pt, minimum width=1.5cm, minimum height=6.12mm, yshift=-0.0mm]
\tikzstyle{white dot}=[fill=white, draw=black, shape=circle]
\tikzstyle{gray dot}=[fill={rgb,255: red,148; green,148; blue,148}, draw=black, shape=circle]
\tikzstyle{dot}=[inner sep=0mm, minimum width=1mm, minimum height=1mm, draw, shape=circle, fill=black]
\tikzstyle{empty dot}=[inner sep=0mm, minimum width=3mm, minimum height=3mm, draw, shape=circle, fill=none]

\tikzstyle{dashed line}=[-, dashed, dash pattern=on 1mm off 0.5mm]
\tikzstyle{dotted line}=[-, style=dotted, tikzit draw=brown]
\tikzstyle{cyan dashed line}=[-, style=dashed, tikzit draw=cyan]
\tikzstyle{green fill line}=[-, fill={green!90!black}, tikzit draw=green]
\tikzstyle{blue fill}=[-, fill=blue, tikzit fill=blue, tikzit draw={rgb,255: red,102; green,117; blue,255}]
\tikzstyle{red}=[-, draw=red, tikzit draw=red]
\tikzstyle{blue}=[-, draw=blue, tikzit draw=blue]
\tikzstyle{thick black}=[-, draw=black, tikzit draw=black, line width=1.8pt]
\tikzstyle{dotted red}=[-, draw=red, style=dotted, tikzit draw=red]
\tikzstyle{dotted blue}=[-, draw=blue, tikzit draw=blue, style=dotted]
\tikzstyle{dashed thick blue}=[-, draw={rgb,255: red,28; green,176; blue,255}, tikzit draw={rgb,255: red,83; green,19; blue,156}, line width=1pt, style=dashed]
\tikzstyle{dashed thick red}=[-, draw=red, tikzit draw={rgb,255: red,255; green,100; blue,10}, line width=1pt, style=dashed]
\tikzstyle{green}=[-, draw=green, tikzit draw=green]
\tikzstyle{dotted green}=[-, draw=green, tikzit draw=green, style=dotted]
\tikzstyle{arrow}=[->]
\tikzstyle{arrow green dashed}=[draw=green, ->, tikzit draw=green, style=dashed]
\tikzstyle{arrow dashed red}=[draw=red, ->, style=dashed, tikzit draw=red]
\tikzstyle{dashed green}=[-, tikzit draw=green, draw=green, style=dashed]

\input{diagrams.tikzdefs}

\DeclareMathOperator*{\argmin}{arg\,min}
\newcommand{\transf}[1]{\ensuremath{\mathcal{#1}}}
\newcommand{\tA}{\transf A}

\newcommand{\tC}{\transf C}
\newcommand{\tD}{\transf D}

\newcommand{\tM}{\transf M}

\newcommand{\tS}{\transf S}

\newcommand{\tU}{\transf U}
\newcommand{\tV}{\transf V}

\newcommand{\tI}{\transf I}

\newcommand{\sys}[1]{\ensuremath{\mathrm{#1}}}
\newcommand{\rA}{\sys A}
\newcommand{\rB}{\sys B}
\newcommand{\rC}{\sys C}

\newcommand{\rE}{\sys E}

\newcommand{\st}{\mathsf{St}}

\newcommand{\vertiii}[1]{{\left\vert\kern-0.25ex\left\vert\kern-0.25ex\left\vert #1
    \right\vert\kern-0.25ex\right\vert\kern-0.25ex\right\vert}}

\newcommand{\cL}{\mathcal{L}}

\newcommand{\cU}{\mathcal{U}}
\newcommand{\cV}{\mathcal{V}}

\newcommand{\Tr}{\mathrm{Tr}}

\def\>{\rangle}

\def\<{\langle}

\def\Tr{\operatorname{Tr}}

\newcommand{\QCs}[1]{\mathfrak{C}(#1)}

\newtheorem{theorem}{Theorem}[section]
\newtheorem{lemma}{Lemma}[section]
\newtheorem{proposition}{Proposition}[section]
\newtheorem{remark}{Remark}

\newtheorem{definition}{Definition}[section]

\begin{document}
\title{Disentangling signalling and causal influence}

\author{Kathleen Barsse}

\email{kathleen.barsse@ens-paris-saclay.fr}

\affiliation{Université Paris-Saclay, ENS Paris-Saclay, 91190,
  Gif-sur-Yvette, and Université de Lorraine, CNRS, Inria, LORIA, F-54000 Nancy, France}

\author{Paolo Perinotti}

\email{paolo.perinotti@unipv.it}

\affiliation{QUIT group, Physics Dept., Pavia University, and INFN Sezione di Pavia, via Bassi 6, 27100 Pavia, Italy}

\author{Alessandro Tosini}

\email{alessandro.tosini@unipv.it}

\affiliation{QUIT group, Physics Dept., Pavia University, and INFN Sezione di Pavia, via Bassi 6, 27100 Pavia, Italy}

\author{Leonardo Vaglini}

\email{leonardo.vaglini01@universitadipavia.it}

\affiliation{Aix-Marseille University, CNRS, LIS, 13288 Marseille CEDEX 09, France}

\begin{abstract}
  The causal effects activated by a quantum interaction are studied,
  modelling the last one as a bipartite unitary channel. The two
  parties, say Alice and Bob, can use the channel to exchange messages---i.e.~to signal. 
  On the other hand, the most general form of causal influence includes also the 
  possibility for Alice, via
  a local operation on her system, to modify Bob's correlations and viceversa. The 
  presence or absence of these two effects are equivalent, but when they both occur, 
  they can differ in their magnitude.
  
  We study the properties of two functions that quantify the amount of signalling
  and causal influence conveyed by an arbitrary unitary
  channel. The functions are proved to be continuous and monotonically
  increasing with respect to the tensor product of channels. Monotonicity
  is instead disproved in the case of sequential composition.

  Signalling and causal influence are analytically computed for the
  quantum SWAP and CNOT gates, in the single use scenario, in
  the $n$-parallel uses scenario, and in the asymptotic regime. A finite gap is 
  found between signalling and causal influence for the quantum CNOT, thus proving
  the existence of extra causal effects that cannot be explained in terms of
  communication only. However, the gap disappears in the asymptotic limit
  of an infinite number of parallel uses, leaving room for asymptotic
  equivalence between signalling and causal influence.

\end{abstract}

\maketitle

\section{Introduction}
The study of causal relations in networks of systems and processes~\cite{pearl_2009,PhysicsPhysiqueFizika.1.195} is a relevant question common to all scientific disciplines. In recent years, much effort in the research on quantum information and foundations has been devoted to the study of quantum causal structures~\cite{Popescu1994,Beckman:2001aa,Bennett-2003,Eggeling:2002aa,PhysRevA.74.012305,Schumacher:2005aa,PhysRevA.72.062323,GutoskiWatrous-testers,PhysRevA.80.022339,PhysRevLett.101.060401}.
On the one hand, motivated by the longstanding problem of the unification of quantum mechanics and general relativity, some authors have considered quantum processes where operations are performed without a definite causal order~\cite{Hardy2009,PhysRevA.88.022318,Oreshkov2012aa,Brukner-qc,Oreshkov_2016,7867830,Perinotti2017,bisio2019theoretical,10.1145/3581760,arrighi2022quantum}. Treating the order in which the operations are performed as a quantum variable also provides an advantage in several informational tasks~\cite{PhysRevLett.113.250402,Milz2022resourcetheoryof,PhysRevLett.117.100502,Chiribella_2021,PhysRevA.86.040301}. On the other hand, a fully quantum version of causal models has been developed, motivated by the challenges of quantum nonlocality~\cite{Costa_2016,PhysRevX.7.031021,barrett2020quantumcausalmodels,BarrettNature2021}, e.g. the impossibility of causally explaining Bell-violating correlations within the framework of classical causal models without relying on some form of fine tuning~\cite{Wood_2015,PhysRevX.8.021018,Pearl2021classicalcausal}.

In the simplest communication setting, there are two parties, say Alice
and Bob, whose systems interact via a unitary channel, and the question boils down to identifying whether a cause-effect relation is induced between the input system of Alice and the output system of Bob and/or viceversa. A causal relation, which has been extensively studied in the literature~\cite{Beckman:2001aa,Eggeling:2002aa,Schumacher:2005aa}, refers to the ability of a party, say Alice, to communicate classical information to the other party, say Bob---a situation that is commonly referred to as \emph{signalling}. This characterisation makes it clear that a necessary condition for Alice to have a causal influence on Bob's system is the presence of an interaction. But when the two systems actually interact, an operation on the input on one side can also affect the correlations between the two parties at the outputs, right after the systems have been left to evolve via the unitary transformation representing the interaction.
The latter possibility---referred to as \emph{causal influence}---has been recently investigated and compared to the signalling condition~\cite{Perinotti2021causalinfluencein}. In quantum theory they have been proved to be equivalent,
in the following sense: a unitary gate does not allow Alice to signal to Bob if and only if it does not induce any causal influence from Alice to Bob.
This equivalence is not an incident of quantum theory, but a phenomenon that can be understood as a consequence of a property called \emph{no interaction without disturbance}. This property states that, whenever discarding Alice's system after the unitary is equivalent to Alice tracing out her system from the very beginning and Bob's system remaining unperturbed, then the unitary acts non-trivially on the Alice's system only---and consequently there is no interaction. It has been shown that this property guarantees coincidence of the no-signalling and no-causal influence conditions~\cite{Perinotti2021causalinfluencein}. The latter statement is thus true in all probabilistic theories that abide by it, e.g. Fermionic Theory and Real Quantum Theory. However, this is not a general fact, in the sense that there exist probabilistic theories where keeping the distinction is mandatory. The counterexample comes from classical information theory where the CNOT is no-signalling from the target to the control, and yet it mediates causal influence in the same direction~\cite{Perinotti2021causalinfluencein}, showing that these two causal relations are manifestly different.

To get a deeper comprehension of the origin of inequivalent causal effects one has to go beyond the qualitative picture described above, where only `yes' or `no' questions can be posed.
This is done introducing functions that quantify the amount of signalling a channel allows for~\cite{barsse2024causalinfluenceversussignalling,goswami2024maximumminimumcausaleffects}, and similarly for causal influence
\cite{barsse2024causalinfluenceversussignalling}. In~\cite{barsse2024causalinfluenceversussignalling}
the authors introduced two measures, $\Sigma(\tU)$ and $C(\tU)$,
defined on the set of unitary channels, called the signalling and
causal influence of $\cU$ respectively, which allowed to more closely study the relation between these two concepts in quantum theory. In
particular, there the following chain of inequalities was proved:
$\Sigma(\cU)\leq C(\cU) \leq2\sqrt{2}\Sigma(\cU)^{1/2}$. A remarkable
corollary is a continuity theorem, which states that if a unitary
channel activates ``small'' signaling, then the total causal influence
must also be small and viceversa. This indicates that in quantum
theory the equivalence between signalling and causal
influence still holds if instead of being null, they are both activated but they are
small. However, the non-linear bound $C(\cU) \leq2\sqrt{2}\Sigma(\cU)^{1/2}$ leaves room for
discrepancies between the magnitude of the two effects even in quantum theory.

To understand whether an actual difference between the two relations can also be observed in the quantum case, it is crucial to explicitly compute $C(\tU)$ and $\Sigma(\tU)$ for some cases of $\tU$.

In this paper we pursue the study of the two measures defined
in~\cite{barsse2024causalinfluenceversussignalling} to provide
definitive evidence of a gap between signalling and causal
influence in quantum theory.
In section \ref{sec:mainprop}, after reviewing the quantifiers of
signalling and causal influence, we prove that they both
satisfy i) continuity as functions on the set of
  unitary channels, and ii) monotonicity with respect to
  tensorization.
  In section \ref{sec:examples} we explicitly evaluate signalling and
  causal influence on two prototype unitary gates: the SWAP and the
  CNOT gate, both in the single shot and asymptotic (tensor product of
  an infinite number of copies) scenario. The SWAP is paradigmatic
  since it represents the most signalling channel among those having
  the same input and output systems. This is somewhat intuitive, in
  that it is an exchange of systems between the two parties. The CNOT
  gate is also of particular interest. In quantum theory, given its
  bidirectionality---control and target can be reversed by means of
  local Hadamard gates---it also allows for signalling from the target
  to the control. We discover that while its causal influence is
  the maximum that is achievable, the signalling is not, and is instead
  smaller, thus proving a gap between the two causal effects. In this
  section we also observe that in the asymptotic regime, the gap
  vanishes for the gates here considered, leaving room for asymptotic
  equivalence in general. 
  Finally, in section \ref{sec:conclusions} we draw conclusions and
discuss open problems.

\textbf{Notation.} In the following, we will use the notation $\rA,\rB
...$ to denote quantum systems. The underlying Hilbert space
corresponding to system $\rA$ will be written as $\mathcal{H}_\rA$, and
its dimension $d_\rA$. The composite system consisting of $\rA$ and $\rB$
will be denoted as $\rA\rB$, with Hilbert space
$\mathcal{H}_{\rA\rB}:=\mathcal{H}_\rA\otimes \mathcal{H}_\rB$. The set of
states over system $\rA$ (in their density matrix representation) will
be written as $\st(\rA)$. The set of quantum channels, i.e. completely
positive trace preserving maps, from $\rA$ to $\rB$ will be written as
$\mathfrak{C}(\rA,\rB)$. When the input and output systems are equal, we
will simply write $\mathfrak{C}(\rA)$. Let $\mathfrak{U}(\rA)$ denote the
set of unitary channels over $\rA$, that is, channels of the form $\rho
\mapsto U \rho U^{\dag}$ for some unitary matrix $U$. Finally, for any
state $\rho\in\st(\rA\rB)$, its marginal sate on system $\rA$ will be written as $\rho_\rA\coloneqq\Tr_\rB(\rho)$.
\section{Quantifying signalling and causal influence}\label{sec:mainprop}

The properties of signalling and causal influence were first defined in a discrete
way (a channel is either signalling or no-signalling, and does or does
not have causal influence)~\cite{Perinotti2021causalinfluencein}.  The
scenario, in a communication setting, consists of two parties, say
Alice and Bob, who share a unitary gate which represents the
interaction between them, and one is interested in establishing
whether a causal relation is induced between the Alice's input and
Bob's output. This situation can be pictorially represented as follows
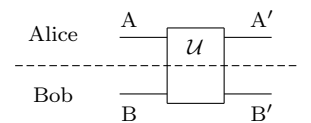
\begin{align*}
\scalebox{1}{\begin{tikzpicture}
	\begin{pgfonlayer}{nodelayer}
		\node [style=none] (1) at (-0.75, 1.25) {};
		\node [style=none] (2) at (-2, 1.25) {};
		\node [style=none] (3) at (-0.75, -0.25) {};
		\node [style=none] (4) at (-2, -0.25) {};
		\node [style=none] (5) at (0.75, -0.25) {};
		\node [style=none] (6) at (2, -0.25) {};
		\node [style=none] (7) at (0.75, 1.25) {};
		\node [style=none] (8) at (2, 1.25) {};
		\node [style=none] (9) at (2.75, 0.5) {};
		\node [style=none] (10) at (-4.75, 0.5) {};
		\node [style=none] (11) at (-0.75, 1.5) {};
		\node [style=none] (12) at (0.75, 1.5) {};
		\node [style=none] (13) at (0.75, -0.5) {};
		\node [style=none] (14) at (-0.75, -0.5) {};
		\node [style=texthere] (15) at (-3.75, 1.25) {Alice};
		\node [style=texthere] (16) at (-3.75, -0.375) {Bob};
		\node [style=texthere] (18) at (-1.75, 1.625) {$\rA$};
		\node [style=texthere] (19) at (-1.75, -0.875) {$\rB$};
		\node [style=texthere] (20) at (1.75, 1.625) {$\rA'$};
		\node [style=texthere] (21) at (1.75, -0.875) {$\rB'$};
		\node [style=texthere] (22) at (0, 0.9) {$\mathcal{U}$};
	\end{pgfonlayer}
	\begin{pgfonlayer}{edgelayer}
		\draw (2.center) to (1.center);
		\draw (7.center) to (8.center);
		\draw (5.center) to (6.center);
		\draw (4.center) to (3.center);
		\draw [style=dashed line] (10.center) to (9.center);
		\draw (11.center) to (12.center);
		\draw (11.center) to (14.center);
		\draw (13.center) to (14.center);
		\draw (13.center) to (12.center);
	\end{pgfonlayer}
\end{tikzpicture}},
\end{align*}
where $\mathcal{U}$ is a bipartite unitary channel shared between
Alice and Bob, $\rA$ and $\rB$ denote Alice and Bob input systems, while
$\rA'$ and $\rB'$ their output systems, respectively. In this setting, much of the focus on causal
effects of unitary channels has been on the possibility for Alice to
use them to communicate to Bob by varying her input state,
i.e. on establishing if the unitary $\tU$ is signalling from Alice to
Bob (or viceversa). The condition for signalling is expressed by negation.

\begin{definition}[no-signalling]\label{def:nos}
  A unitary channel
  $\tU\in\mathfrak{C}(\rA\rB,\rA'\rB')$ is \emph{no-signaling} from 
$\rA$ to $\rB'$ if
\begin{equation}
\scalebox{1}{\begin{tikzpicture}
	\begin{pgfonlayer}{nodelayer}
		\node [style=none] (0) at (-2, 0.75) {};
		\node [style=none] (1) at (-2, -0.75) {};
		\node [style=none] (9) at (-0.75, 0.75) {};
		\node [style=none] (10) at (-0.75, -0.75) {};
		\node [style=none] (12) at (0.75, 0.75) {};
		\node [style=none] (13) at (2, 0.75) {};
		\node [style=none] (14) at (0.75, -0.75) {};
		\node [style=none] (15) at (2.75, -0.75) {};
		\node [style=medium box] (16) at (0, 0) {$\mathcal{U}$};
		\node [style=upground] (17) at (2.2, 0.75) {};
		\node [style=texthere] (19) at (1.5, 1.075) {$\rA'$};
		\node [style=texthere] (20) at (2.625, -0.425) {$\rB'$};
		\node [style=texthere] (21) at (-1.75, 1.075) {$\rA$};
		\node [style=texthere] (22) at (-1.75, -0.425) {$\rB$};
	\end{pgfonlayer}
	\begin{pgfonlayer}{edgelayer}
		\draw (0.center) to (9.center);
		\draw (1.center) to (10.center);
		\draw (12.center) to (13.center);
		\draw (14.center) to (15.center);
	\end{pgfonlayer}
\end{tikzpicture}} \ = \ \scalebox{1}{\begin{tikzpicture}
	\begin{pgfonlayer}{nodelayer}
		\node [style=none] (23) at (-1.75, 0.75) {};
		\node [style=none] (24) at (-1.75, -0.75) {};
		\node [style=none] (26) at (-0.5, 0.75) {};
		\node [style=none] (27) at (-0.5, -0.75) {};
		\node [style=none] (28) at (0.5, -0.75) {};
		\node [style=none] (29) at (1.75, -0.75) {};
		\node [style=upground] (30) at (-0.3, 0.75) {};
		\node [style=texthere] (31) at (-1.5, 1.075) {$\rA$};
		\node [style=texthere] (32) at (-1.5, -0.425) {$\rB$};
		\node [style=texthere] (33) at (1.5, -0.425) {$\rB'$};
		\node [style=tiny box] (34) at (0, -0.75) {$\mathcal{C}$};
	\end{pgfonlayer}
	\begin{pgfonlayer}{edgelayer}
		\draw (23.center) to (26.center);
		\draw (24.center) to (27.center);
		\draw (28.center) to (29.center);
	\end{pgfonlayer}
\end{tikzpicture}},
\label{eq:nos}
\end{equation}
where the symbol at the output of systems $\rA$ and $\rA'$ denotes the
trace operator, otherwise it is signalling from $\rA$ to $\rB'$. The channel is \emph{signalling} if it is not no-signalling.
\end{definition}

Indeed a bipartite channel can be used to signal from Alice to
Bob if and only if some local intervention on Alice's system can modify the outcome
probabilities of some local measurement of Bob's system, which is indeed forbidden
by the above condition~\eqref{eq:nos}. However, the notion of signalling
does not exhaust all the ways in which Alice can have a causal
influence on Bob.  Indeed, a local intervention of Alice could be
propagated by the evolution $\tU$ modifying the correlations between Bob
and Alice's output systems, and yet satisfying the criterion
in Eq.~\eqref{eq:nos}. This is why we introduce the next notion, that can be summarised 
by saying that a unitary evolution involving both Alice and Bob's systems does not
induce causal influence from $\rA$ to $\rB'$ if every local intervention of Alice remains 
localised on Alice's side.

\begin{definition}[no-causal influence]\label{def:noci}
A unitary channel $\tU\in\mathfrak{C}(\rA\rB,\rA'\rB')$
\emph{has no causal influence from $\rA$ to $\rB'$} if  for every $
\tA\in\QCs{\text{EA}}$ one has
\begin{equation}
 \begin{aligned}
    \Qcircuit @C=1em @R=1.5em
    {    &\ustick{\rE} & \qw & \qw & \multigate{1}{\tA} & \qw &\qw & \ustick{\rE}\qw \\
	&\ustick{\rA'}&\multigate{1}{\cU^{-1}}&\ustick{\rA}\qw&\ghost{\tA}&\ustick{\rA}\qw&\multigate{1}{\cU}&\ustick{\rA'}\qw\\
      &\ustick{\rB'}&\ghost{\cU^{-1}}&\qw&\ustick{\rB}\qw&\qw&\ghost{\cU}&\ustick{\rB'}\qw}
  \end{aligned}
\quad =
\begin{aligned}	
\Qcircuit @C=1em @R=1.5em
	{	&\ustick{\rE} & \multigate{1}{\tA'} & \ustick{\rE}\qw \\
		&\ustick{\rA'}&\ghost{\tA'}&\ustick{\rA'}\qw\\
		&&\ustick{\rB'}\qw&\qw
	}
	\end{aligned}\quad,
\label{eq:noci}
\end{equation}
where $\rE$ is an arbitrary ancillary system. A unitary channel \emph{has causal influence from $\rA$ to $\rB'$} if the above condition fails.
\end{definition}

Notice that in the definition we also include the possibility for Alice
to act on an extended system $\rA\rE$, for some arbitrary ancilla
$\rE$. The condition guarantees that the channel $\tU$ does not
influence Bob's correlations with any system, including ancillary
ones.

\subsection{Definition of the functions $\Sigma(\tU)$ and $C(\tU)$ on
  the space of unitary quantum channels}
\label{section:definitions}

In Ref~\cite{barsse2024causalinfluenceversussignalling}, precise quantifiers
of signalling and causal influence have been
introduced. Then, if a channel is signalling, we can precisely
quantify how much that channel can signal, and similarly for causal
influence.  We begin this section by reviewing the definition of
signalling and causal influence given
in~\cite{barsse2024causalinfluenceversussignalling}. The definitions
are directly based on their discrete versions: the closer we are
to the conditions in Eqs~\eqref{eq:nos} and~\eqref{eq:noci} the
less signalling or causal influence is activated by the channel, respectively.

To assess the amount of signalling, we look for the channel $\mathfrak{C}(B,B')$ that
best gauges the distance between the l.h.s. and the r.h.s.~of Eq.~\eqref{eq:nos}. The
distance is taken via the diamond norm, which has a clear
operational interpretation in terms of the success probability in
discriminating the channels.

\begin{definition}[signalling of $\tU$]
\label{def:signalling}
 Given any bipartite unitary channel $\mathcal{U}$, its
 \emph{signalling} from $A$ to $B'$ is quantified via the following function
 \begin{align*}
&  \Sigma:\mathfrak{U}(\rA\rA',\rB,\rB')\rightarrow \mathbb{R},\qquad  \tU \mapsto\Sigma(\tU)\coloneqq\\
&  \inf_{\mathcal{C}\in \mathfrak{C}(\rB,\rB')}\left\|\scalebox{1}{\begin{tikzpicture}
	\begin{pgfonlayer}{nodelayer}
		\node [style=none] (0) at (-2, 0.75) {};
		\node [style=none] (1) at (-2, -0.75) {};
		\node [style=none] (9) at (-0.75, 0.75) {};
		\node [style=none] (10) at (-0.75, -0.75) {};
		\node [style=none] (12) at (0.75, 0.75) {};
		\node [style=none] (13) at (2, 0.75) {};
		\node [style=none] (14) at (0.75, -0.75) {};
		\node [style=none] (15) at (2.75, -0.75) {};
		\node [style=medium box] (16) at (0, 0) {$\mathcal{U}$};
		\node [style=upground] (17) at (2.2, 0.75) {};
		\node [style=texthere] (19) at (1.5, 1.075) {$\rA'$};
		\node [style=texthere] (20) at (2.625, -0.425) {$\rB'$};
		\node [style=texthere] (21) at (-1.75, 1.075) {$\rA$};
		\node [style=texthere] (22) at (-1.75, -0.425) {$\rB$};
	\end{pgfonlayer}
	\begin{pgfonlayer}{edgelayer}
		\draw (0.center) to (9.center);
		\draw (1.center) to (10.center);
		\draw (12.center) to (13.center);
		\draw (14.center) to (15.center);
	\end{pgfonlayer}
\end{tikzpicture}}  -  \scalebox{1}{\begin{tikzpicture}
	\begin{pgfonlayer}{nodelayer}
		\node [style=none] (23) at (-1.75, 0.75) {};
		\node [style=none] (24) at (-1.75, -0.75) {};
		\node [style=none] (26) at (-0.5, 0.75) {};
		\node [style=none] (27) at (-0.5, -0.75) {};
		\node [style=none] (28) at (0.5, -0.75) {};
		\node [style=none] (29) at (1.75, -0.75) {};
		\node [style=upground] (30) at (-0.3, 0.75) {};
		\node [style=texthere] (31) at (-1.5, 1.075) {$\rA$};
		\node [style=texthere] (32) at (-1.5, -0.425) {$\rB$};
		\node [style=texthere] (33) at (1.5, -0.425) {$\rB'$};
		\node [style=tiny box] (34) at (0, -0.75) {$\mathcal{C}$};
	\end{pgfonlayer}
	\begin{pgfonlayer}{edgelayer}
		\draw (23.center) to (26.center);
		\draw (24.center) to (27.center);
		\draw (28.center) to (29.center);
	\end{pgfonlayer}
\end{tikzpicture}}\ \right\|_{\diamond}.
 \end{align*}
\end{definition}

We will show in the next subsection that the infimum in the definition
can be replaced with a minimum. A unitary channel is then
no-signalling if and only if $\Sigma(\tU)=0$, in agreement with
Definition~\ref{def:nos}.

We now focus on causal influence. The condition expressed by
\eqref{eq:noci} is impractical since it requires to check an
infinite number of interventions $\tA$. However, this issue was circumvented
in~\cite{Perinotti2021causalinfluencein} where it is shown that one
needs to check what happens with the SWAP channel only, i.e. $\tU$ has no-causal
influence from $\rA$ to $\rB'$ if and only if the following equation
holds
\begin{equation}\label{eq:nocibis}
\begin{tikzpicture}
	\begin{pgfonlayer}{nodelayer}
		\node [style=medium box] (0) at (-2.125, -0.75) {$\mathcal{U}^{-1}$};
		\node [style=medium box] (1) at (2.125, -0.75) {$\mathcal{U}$};
		\node [style=none] (2) at (-0.875, 0) {};
		\node [style=none] (3) at (-2.875, 0) {};
		\node [style=none] (4) at (1.375, -1.5) {};
		\node [style=none] (5) at (-2.875, -1.5) {};
		\node [style=none] (6) at (-1.375, -1.5) {};
		\node [style=none] (7) at (0.875, 0) {};
		\node [style=none] (8) at (2.875, -1.5) {};
		\node [style=none] (9) at (2.875, 0) {};
		\node [style=none] (10) at (-0.875, 1.5) {};
		\node [style=none] (11) at (0.875, 1.5) {};
		\node [style=none] (12) at (4.125, 0) {};
		\node [style=none] (13) at (4.125, -1.5) {};
		\node [style=none] (14) at (4.125, 1.5) {};
		\node [style=none] (15) at (-3.875, 1.5) {};
		\node [style=none] (16) at (-3.875, 0) {};
		\node [style=none] (17) at (-3.875, -1.5) {};
		\node [style=none] (18) at (-1.375, 0) {};
		\node [style=none] (20) at (1.375, 0) {};
		\node [style=texthere] (21) at (4, 1.825) {$\rA$};
		\node [style=texthere] (22) at (4, 0.325) {$\rA'$};
		\node [style=texthere] (23) at (4, -1.175) {$\rB'$};
		\node [style=texthere] (24) at (0.875, 0.325) {$\rA$};
		\node [style=texthere] (25) at (0.875, -1.175) {$\rB$};
		\node [style=texthere] (26) at (-0.875, 0.325) {$\rA$};
		\node [style=texthere] (27) at (-3.75, 1.825) {$\rA$};
		\node [style=texthere] (28) at (-3.75, 0.325) {$\rA'$};
		\node [style=texthere] (29) at (-3.75, -1.175) {$\rB'$};
	\end{pgfonlayer}
	\begin{pgfonlayer}{edgelayer}
		\draw (15.center) to (10.center);
		\draw (16.center) to (3.center);
		\draw (17.center) to (5.center);
		\draw (6.center) to (4.center);
		\draw [in=180, out=0, looseness=1.25] (2.center) to (11.center);
		\draw [in=-180, out=0, looseness=1.25] (10.center) to (7.center);
		\draw (11.center) to (14.center);
		\draw (9.center) to (12.center);
		\draw (8.center) to (13.center);
		\draw (7.center) to (20.center);
		\draw (18.center) to (2.center);
	\end{pgfonlayer}
\end{tikzpicture} \ = \ \begin{tikzpicture}
	\begin{pgfonlayer}{nodelayer}
		\node [style=none] (3) at (-0.75, 0) {};
		\node [style=none] (5) at (2, -1.5) {};
		\node [style=none] (9) at (0.75, 0) {};
		\node [style=none] (10) at (-0.75, 1.5) {};
		\node [style=none] (11) at (0.75, 1.5) {};
		\node [style=none] (12) at (2, 0) {};
		\node [style=none] (14) at (2, 1.5) {};
		\node [style=none] (15) at (-2, 1.5) {};
		\node [style=none] (16) at (-2, 0) {};
		\node [style=none] (17) at (-2, -1.5) {};
		\node [style=texthere] (20) at (1.875, 1.825) {$\rA$};
		\node [style=texthere] (21) at (1.875, 0.325) {$\rA'$};
		\node [style=texthere] (26) at (-1.75, 1.825) {$\rA$};
		\node [style=texthere] (27) at (-1.75, 0.325) {$\rA'$};
		\node [style=texthere] (28) at (1.875, -1.175) {$\rB'$};
		\node [style=none] (32) at (-2, 0) {};
		\node [style=none] (33) at (-2, -1.5) {};
		\node [style=none] (34) at (-2, 1.5) {};
		\node [style=medium box] (40) at (0, 0.75) {$\mathcal{C}$};
	\end{pgfonlayer}
	\begin{pgfonlayer}{edgelayer}
		\draw (15.center) to (10.center);
		\draw (16.center) to (3.center);
		\draw (17.center) to (5.center);
		\draw (11.center) to (14.center);
		\draw (9.center) to (12.center);
	\end{pgfonlayer}
\end{tikzpicture}.
\end{equation}

Based on this equivalence, the causal influence is defined in an
analogous way to the signalling as follows:
\begin{definition}[causal influence of $\tU$]
\label{def:causal_influence}
 Given a bipartite unitary channel $\mathcal{U}$, the causal influence
 of $\mathcal{U}$ from $\rA$ to $\rB'$ is measured  by the following function:
\onecolumn
 \begin{align*}
  C:&\mathfrak{U}(\rA\rA',\rB,\rB')\rightarrow \mathbb{R},\\
  &\tU \mapsto C(\tU) \coloneqq\inf_{\mathcal{C}\in \mathfrak{C}(\rA\rA')} 
                                          \left\|\scalebox{1}{\begin{tikzpicture}
	\begin{pgfonlayer}{nodelayer}
		\node [style=medium box] (0) at (-2.125, -0.75) {$\mathcal{U}^{-1}$};
		\node [style=medium box] (1) at (2.125, -0.75) {$\mathcal{U}$};
		\node [style=none] (2) at (-0.875, 0) {};
		\node [style=none] (3) at (-2.875, 0) {};
		\node [style=none] (4) at (1.375, -1.5) {};
		\node [style=none] (5) at (-2.875, -1.5) {};
		\node [style=none] (6) at (-1.375, -1.5) {};
		\node [style=none] (7) at (0.875, 0) {};
		\node [style=none] (8) at (2.875, -1.5) {};
		\node [style=none] (9) at (2.875, 0) {};
		\node [style=none] (10) at (-0.875, 1.5) {};
		\node [style=none] (11) at (0.875, 1.5) {};
		\node [style=none] (12) at (4.125, 0) {};
		\node [style=none] (13) at (4.125, -1.5) {};
		\node [style=none] (14) at (4.125, 1.5) {};
		\node [style=none] (15) at (-4.125, 1.5) {};
		\node [style=none] (16) at (-4.125, 0) {};
		\node [style=none] (17) at (-4.125, -1.5) {};
		\node [style=none] (18) at (-1.375, 0) {};
		\node [style=none] (19) at (1.375, 0) {};
		\node [style=texthere] (20) at (4, 1.825) {$\rA$};
		\node [style=texthere] (21) at (4, 0.325) {$\rA'$};
		\node [style=texthere] (22) at (4, -1.175) {$\rB'$};
		\node [style=texthere] (23) at (0.875, 0.325) {$\rA$};
		\node [style=texthere] (24) at (0.875, -1.175) {$\rB$};
		\node [style=texthere] (25) at (-0.875, 0.325) {$\rA$};
		\node [style=texthere] (26) at (-3.875, 1.825) {$\rA$};
		\node [style=texthere] (27) at (-3.875, 0.325) {$\rA'$};
		\node [style=texthere] (28) at (-3.875, -1.175) {$\rB'$};
		\node [style=none] (33) at (-4.125, 0) {};
		\node [style=none] (34) at (-4.125, -1.5) {};
		\node [style=none] (35) at (-4.125, 1.5) {};
	\end{pgfonlayer}
	\begin{pgfonlayer}{edgelayer}
		\draw (15.center) to (10.center);
		\draw (16.center) to (3.center);
		\draw (17.center) to (5.center);
		\draw (6.center) to (4.center);
		\draw [in=180, out=0, looseness=1.25] (2.center) to (11.center);
		\draw [in=-180, out=0, looseness=1.25] (10.center) to (7.center);
		\draw (11.center) to (14.center);
		\draw (9.center) to (12.center);
		\draw (8.center) to (13.center);
		\draw (7.center) to (19.center);
		\draw (18.center) to (2.center);
	\end{pgfonlayer}
\end{tikzpicture}} -\scalebox{1}{\begin{tikzpicture}
	\begin{pgfonlayer}{nodelayer}
		\node [style=none] (3) at (-0.75, 0) {};
		\node [style=none] (5) at (2, -1.5) {};
		\node [style=none] (9) at (0.75, 0) {};
		\node [style=none] (10) at (-0.75, 1.5) {};
		\node [style=none] (11) at (0.75, 1.5) {};
		\node [style=none] (12) at (2, 0) {};
		\node [style=none] (14) at (2, 1.5) {};
		\node [style=none] (15) at (-2, 1.5) {};
		\node [style=none] (16) at (-2, 0) {};
		\node [style=none] (17) at (-2, -1.5) {};
		\node [style=texthere] (20) at (1.875, 1.825) {$\rA$};
		\node [style=texthere] (21) at (1.875, 0.325) {$\rA'$};
		\node [style=texthere] (26) at (-1.75, 1.825) {$\rA$};
		\node [style=texthere] (27) at (-1.75, 0.325) {$\rA'$};
		\node [style=texthere] (28) at (1.875, -1.175) {$\rB'$};
		\node [style=none] (32) at (-2, 0) {};
		\node [style=none] (33) at (-2, -1.5) {};
		\node [style=none] (34) at (-2, 1.5) {};
		\node [style=medium box] (40) at (0, 0.75) {$\mathcal{C}$};
	\end{pgfonlayer}
	\begin{pgfonlayer}{edgelayer}
		\draw (15.center) to (10.center);
		\draw (16.center) to (3.center);
		\draw (17.center) to (5.center);
		\draw (11.center) to (14.center);
		\draw (9.center) to (12.center);
	\end{pgfonlayer}
\end{tikzpicture}} \right\|_{\diamond}.
 \end{align*}
\twocolumn
\end{definition}

The above function $C$ evaluates how close a channel $\tU$ is to
satisfying the no-causal influence condition of
Definition~\ref{def:noci}, with no-causal influence corresponding to
$\mathcal{C}(\mathcal{U})=0$. Also in this case we will prove that
the infimum is actually the minimum.
We observe that by spelling out the definition of diamond norm we can rewrite signalling and causal influence of $\tU$ form $\rA$ to $\rB'$ as
\begin{align}
&  \Sigma(\mathcal{U})=\inf_{\mathcal{C}\in \mathfrak{C}(B,B')}
  \sup_{\rho \in \st(EAB)} \label{eq:sigminimax}\\
&  \qquad\left\| \bigl[\bigl(\mathcal{I}_E \otimes( \Tr_{A'}\otimes \mathcal{I}_{B'})\ \mathcal{U}\bigr) - \bigl(\mathcal{I}_E \otimes \Tr_{A}\otimes \mathcal{C} \bigl)\bigr](\rho) \right\|_1 \nonumber\\
&  \mathcal{C}(\mathcal{U})=\inf_{\mathcal{C}\in \mathfrak{C}(AA')}
  \sup_{\rho \in \st(EAA'B')}\label{eq:cinfminimax}\\
&\qquad  \left\| \bigl[(\mathcal{I}_E\otimes \mathcal{T}(\mathcal{U})) -  (\mathcal{I}_E\otimes \mathcal{C} \otimes \mathcal{I}_{B'})\bigr](\rho)\right\|_1.\nonumber
\end{align}
  
For any pair of quantum states $\sigma$ and $\rho$, the trace norm
$\|\sigma -\rho\|_1$ is twice the trace distance $D(\sigma, \rho)$
between those states, which lies in the interval $[0,1]$. In the
expressions of $\Sigma$ and $C$, the trace norm is evaluated over a
difference of quantum states, therefore we have the following
admissible ranges for the signalling and causal influence of a channel:
\begin{equation}\label{eq:ranges}
0\leq \Sigma(\mathcal{U})\leq 2,\qquad 0\leq C(\mathcal{U})\leq 2.
\end{equation}

\begin{remark}[Choice of diamond norm in $\Sigma(\tU)$ and $C(\tU)$]\emph{The diamond norm is motivated by its
  physical meaning, though other choices like the Hilbert-Schmidt norm could
  have been made. As long as we are dealing with finite dimensional
  systems, all norms are equivalent. However, different norms could
  give rise to different orderings, in the sense that a channel
  could have more signalling (causal influence) than another for one
  norm and viceversa for another norm. Since we are interested in
  comparing the magnitude of causal effects for different channels, it
  is convenient to rely on a distance with a clear operational interpretation.}
\end{remark}

The relation between signalling and causal influence has been studied in Refs.~\cite{Perinotti2021causalinfluencein,barsse2024causalinfluenceversussignalling}. In Ref.~\cite{Perinotti2021causalinfluencein}, using the discrete definition, it was shown that a quantum channel $\mathcal{U}$ is no-signalling if and only if it has no causal influence. 
This can be restated with Ref.~\cite{barsse2024causalinfluenceversussignalling}'s definitions as follows:
\begin{equation}\label{eq:discrete-equivalence}
 \Sigma(\mathcal{U})=0 \Leftrightarrow C(\mathcal{U})=0.
\end{equation}
This result was expanded upon in Ref.~\cite{barsse2024causalinfluenceversussignalling}, where the following bounds between causal influence and signalling were proved:
\begin{equation}
\label{eq:bound_gen}
 \Sigma(\mathcal{U}) \leq C(\mathcal{U})\leq 2\sqrt{2}\Sigma( \mathcal{U})^{\frac{1}{2}}.
\end{equation}
The latter provides the continuity of signalling vs causal influence
that, despite the quadratic scaling predicted in the upper bound, turn
on continuously when an interaction occurs: ``if signalling is small
then causal influence is also small, and viceversa''. The discrete
equivalence in Eq.~\eqref{eq:discrete-equivalence} is thus robust in
the regime of ``small casual effects''. However, the two quantities may
significantly differ in magnitude: the ratio between causal influence and signalling could 
in principle be unbounded even in the regime where both are small. The scope of next
section is to find a signature of the fact that causal influence can exceed signalling. To
further investigate the above relations, in the next section  we prove the main properties
of the two functions in Definitions~\ref{def:signalling} and~\ref{def:causal_influence}.

\subsection{Properties}
\label{section:properties}

In this section, we show the properties of signalling and
causal influence that will be used in the next section for their comparison.  
In the following we will use the useful quantities defined for every unitary $\tU$ in $\mathfrak{U}(\rA\rB,\rA'\rB')$ and every channel $\tC$ in $\mathfrak C(\rB,\rB')$ (resp. $\mathfrak{C}(\rA\rA')$): 
\begin{equation}\label{eq:spelled-defs-2}
  \begin{aligned}
&\Sigma_{\mathcal{C}}(\mathcal{U})\coloneqq  \left\| (\Tr_{\rA'}\otimes \mathcal{I}_{\rB'})\ \mathcal{U} -
                 \Tr_{\rA}\otimes \mathcal{C} \right\|_{\diamond},\\
&  C_{\mathcal{C}}(\mathcal{U})
     \coloneqq\left\|  \mathcal{T}(\mathcal{U}) -  \mathcal{C} \otimes
       \mathcal{I}_{\rB'}\right\|_{\diamond},
   \end{aligned}
   \end{equation}
so that we can write
\begin{align}\label{eq:spelled-defs}
&\Sigma(\mathcal{U})=\inf_{\mathcal{C}\in\mathfrak{C}(\rB,\rB')}
\Sigma_{\mathcal{C}}(\mathcal{U}),\qquad C(\mathcal{U})=\inf_{\mathcal{C}\in\mathfrak{C}(\rA\rA')}C_{\mathcal{C}}(\mathcal{U}).
\end{align}

We first show that the infimum and
supremum in the explicit definitions of $\Sigma$ and $C$ (see
Eqs.~\eqref{eq:spelled-defs} and~\eqref{eq:spelled-defs-2} are in fact a minimum and maximum, with
the latter being achieved on pure states.

\begin{lemma}
 \label{lem:infsup}
 Considering the expressions~\eqref{eq:sigminimax} and~\eqref{eq:cinfminimax}, the following statements hold.
\begin{enumerate}
\item \label{it:minmax}  the infimum and supremum are a minimum and maximum:
 \begin{align*}
   &  \Sigma(\mathcal{U}) = \min_{\mathcal{C}\in \mathfrak{C}(B,B')}
     \max_{\rho \in \st(EAB)} \\
   &\quad \| [(\mathcal{I}_E \otimes( \Tr_{A'}\otimes \mathcal{I}_{B'})\ \mathcal{U}) - (\mathcal{I}_E \otimes \Tr_{A}\otimes \mathcal{C})](\rho) \|_1,\\
   &  C(\mathcal{U}) = \min_{\mathcal{C}\in \mathfrak{C}(AA')} \max_{\rho
     \in \st(EAA'B')} \\
   &\quad\| [(\mathcal{I}_E\otimes \mathcal{T}(\mathcal{U})) -  (\mathcal{I}_E\otimes \mathcal{C} \otimes \mathcal{I}_{B'})](\rho)\|_1;
 \end{align*}
\item\label{it:convex}
the sets 
\begin{align*}
&O_\Sigma(\mathcal U)\coloneqq\argmin_{\mathcal{C}\in
                 \mathfrak{C}(B,B')} \max_{\rho \in \st(EAB)} \\
  &\quad\left\| \bigl[\bigl(\mathcal{I}_E \otimes( \Tr_{A'}\otimes \mathcal{I}_{B'})\ \mathcal{U}\bigr) - \bigl(\mathcal{I}_E \otimes \Tr_{A}\otimes \mathcal{C} \bigl)\bigr](\rho) \right\|_1, \\
&O_C(\mathcal{U}) \coloneqq \argmin_{\mathcal{C}\in \mathfrak{C}(AA')} \max_{\rho \in \st(EAA'B')} \\&\quad\left\| \bigl[(\mathcal{I}_E\otimes \mathcal{T}(\mathcal{U})) -  (\mathcal{I}_E\otimes \mathcal{C} \otimes \mathcal{I}_{B'})\bigr](\rho)\right\|_1,
\end{align*}
of optimal channels are convex;
\item\label{it:purmax}
the above maxima can always be attained with a pure state.
\end{enumerate}
\end{lemma}

\begin{proof}
  \ref{it:minmax}. First, to show that the supremum in
  $\Sigma(\mathcal{U})$ is a maximum, we prove that $\st(EAB)$ is
  compact and the function
  $\rho\mapsto \left\| \bigl[\bigl(\mathcal{I}_E \otimes(
    \Tr_{A'}\otimes \mathcal{I}_{B'})\ \mathcal{U}\bigr) -
    \bigl(\mathcal{I}_E \otimes \Tr_{A}\otimes \mathcal{C}
    \bigl)\bigr](\rho) \right\|_1$ is continuous. Indeed, the set of
  positive semi-definite matrices and set the of matrices of trace 1
  are closed, so their intersection $\st(EAB)$ is closed, and
  $\st(EAB)$ is also bounded. Therefore, $\st(EAB)$ is a compact set
  (as a subset of the finite-dimensional real vector space of
  Hermitian matrices of a given dimension). To prove continuity, we
  use the fact that
  $\rho\mapsto \bigl[\bigl(\mathcal{I}_E \otimes( \Tr_{A'}\otimes
  \mathcal{I}_{B'})\ \mathcal{U}\bigr) - \bigl(\mathcal{I}_E \otimes
  \Tr_{A}\otimes \mathcal{C} \bigl)\bigr](\rho)$ is continuous as a
  linear function over a finite-dimensional space and that
  $\|\cdot\|_{1}$ is continuous. The proof that the supremum in
  $C(\mathcal{U})$ is a maximum is straightforwardly similar.  We now
  prove that the infimum in
  $\inf_{\mathcal{C}\in \mathfrak{C}(B,B')}
  \Sigma_{\mathcal{C}}(\mathcal{U})$ is a minimum. To prove this, we
  show that $\mathfrak{C}(B,B')$ is a compact set and that
  $\mathcal{C}\mapsto \Sigma_{\mathcal{C}}(\mathcal{U})$ is
  continuous. First, the sets of completely positive maps from $B$ to
  $B'$ and of trace-preserving maps from $B$ to $B'$ are closed,
  therefore their intersection, the set of quantum channels, is
  closed. 
  Moreover, since the diamond norm of all quantum channels is 1
  (Lemma~\ref{lem:propdiamondkitaev}), $\mathfrak{C}(B,B')$ is bounded
  and is therefore a compact set (as a subset of a
  finite-dimensional real vector space consisting of
  Hermitian-preserving linear maps). Next, we show that
  $\mathcal{C}\mapsto \Sigma_{\mathcal{C}}(\mathcal{U})$ is
  continuous. For all pair of channels
  $\mathcal{C},\mathcal{D}\in\mathfrak{C}(B,B')$, we have:
\begin{align*}
&   \Sigma_{\mathcal{C}}(\mathcal{U}) =\left\| ( \Tr_{A'}\otimes \mathcal{I}_{B'})\ \mathcal{U} -  \Tr_{A}\otimes \mathcal{C} \right\|_{\diamond} \\
 &\leq  \left\| ( \Tr_{A'}\otimes \mathcal{I}_{B'})\ \mathcal{U} -
         \Tr_{A}\otimes \mathcal{D} \right\|_{\diamond}+\\
&\qquad \qquad \qquad \qquad \qquad \qquad \quad\:\:\left\|  \Tr_{A}\otimes \mathcal{D} -  \Tr_{A}\otimes \mathcal{C} \right\|_{\diamond} \\
& =   \Sigma_{\mathcal{D}}(\mathcal{U}) +\\
 &\qquad\sup_{\rho \in \st(EAB)} \left\| \bigl[ \mathcal{I}_E \otimes \Tr_{A}\otimes \mathcal{D} -  \mathcal{I}_E \otimes \Tr_{A}\otimes \mathcal{C}\bigr]  (\rho)\right\|_1 \\
 &\leq\Sigma_{\mathcal{D}}(\mathcal{U}) +\sup_{\rho \in \st(EAB)} \left\| \bigl[ \mathcal{I}_{EA}\otimes \mathcal{D} -  \mathcal{I}_{EA} \otimes \mathcal{C}\bigr]  (\rho)\right\|_1 \\
&=\Sigma_{\mathcal{D}}(\mathcal{U}) + \|\mathcal{D}-\mathcal{C}\|_{\diamond}.
\end{align*}
In the fourth line, we use the fact that the trace norm is contractive for the partial trace~\cite{nielsen_chuang_2010} and stable under tensor with the identity (see Proposition~\ref{prop:stab_id}). Therefore, $\Sigma_{\mathcal{C}}(\mathcal{U}) -\Sigma_{\mathcal{D}}(\mathcal{U}) \leq  \|\mathcal{D}-\mathcal{C}\|_{\diamond}$. Since the roles of $\mathcal{C}$ and $\mathcal{D}$ are symmetrical, we have $|\Sigma_{\mathcal{C}}(\mathcal{U}) -\Sigma_{\mathcal{D}}(\mathcal{U})| \leq \|\mathcal{D}-\mathcal{C}\|_{\diamond}$. Therefore, $\mathcal{C}\mapsto \Sigma_{\mathcal{C}}(\mathcal{U})$ is 1-Lipschitz continuous (and therefore continuous). The proof that the infimum $\inf_{\mathcal{C}\in \mathfrak{C}(AA')} C_{\mathcal{C}}(\mathcal{U})$ is a minimum is straightforwardly similar.

\ref{it:convex}. Let $\tC$ be a channel and set $\mathcal{L}_\tC\coloneqq( \Tr_{A'}\otimes \mathcal{I}_{B'})\ \mathcal{U} - \Tr_{A}\otimes \mathcal{C} $.  Let $\tC_1$, $\tC_2$ be two optimal channels, so that $\Sigma(\tU)=\|\cL_{\tC_1} \|_{\diamond}=\|\cL_{\tC_2} \|_{\diamond}$. By the triangle inequality, for every $p\in[0,1]$ we have
\[
\| \cL_{p\tC_1+(1-p)\tC_2} \|_{\diamond}\leq p\| \cL_{\tC_1} \|_{\diamond}+(1-p)\| \cL_{\tC_2} \|_{\diamond}=\Sigma(\tU),
\]
therefore also $p \,\tC_1+(1-p)\,\tC_2$ is optimal, as claimed.

\ref{it:purmax}. We finally prove that maxima can be attained with pure states. For the signalling (resp. the causal influence), let $\mathcal{L}\coloneqq\mathcal{I}_E \otimes( \Tr_{A'}\otimes \mathcal{I}_{B'})\ \mathcal{U} - \mathcal{I}_E \otimes \Tr_{A}\otimes \mathcal{C} $ (resp. $\mathcal{L}\coloneqq\mathcal{I}_E\otimes \mathcal{T}(\mathcal{U}) -  \mathcal{I}_E\otimes \mathcal{C} \otimes \mathcal{I}_{B'}$) and let $\rho$ be a state that maximizes $\left\| \mathcal{L}(\rho) \right\|_1 $. $\rho$ can be decomposed into a convex combination of pure states: $\rho=\sum_i \lambda_i \ket{\psi_i}\bra{\psi_i}$, where for all $i$, $\lambda_i>0$ and $\sum_i \lambda_i=1$. Therefore:
 \begin{align*}
   \|\mathcal{L}(\rho)\|_1&=\left\|\sum_i \lambda_i
                            \mathcal{L}(\ket{\psi_i}\bra{\psi_i})\right\|_1 \\ &
                            \leq  \sum_i\lambda_i
                            \|\mathcal{L}(\ket{\psi_i}\bra{\psi_i})\|_1
                                                                                 \leq \sum_i\lambda_i \|\mathcal{L}(\rho)\|_1 \\
                          &= \|\mathcal{L}(\rho)\|_1.
 \end{align*}
 \noindent because $\|\mathcal{L}(\rho)\|_1$ is maximal. Therefore, $\|\mathcal{L}(\ket{\psi_i}\bra{\psi_i})\|$ is also maximal for all $i$.
\end{proof}

We then show that $\Sigma$ and $C$ are continuous functions:
\begin{lemma}[continuity]
 \label{continuity}
 Signalling and causal influence are continuous functions from $\mathfrak{U}(AB,A'B')$ to $[0,2]$.
\end{lemma}
\begin{proof}
For the signalling, consider two unitary channels $\mathcal{U}$ and $\mathcal{V}$ of the same dimension. We have:
\begin{align*} 
 \Sigma(\mathcal{U}) =&\min_{\mathcal{C}\in \mathfrak{C}(\rB,\rB')}  \left\| ( \Tr_{A'}\otimes \mathcal{I}_{\rB'})\ \mathcal{U} -  \Tr_{\rA}\otimes \mathcal{C} \right\|_{\diamond}  \\
 \leq & \min_{\mathcal{C}\in \mathfrak{C}(\rB,\rB')} \bigl[\left\| (
        \Tr_{\rA'}\otimes \mathcal{I}_{\rB'})\ \mathcal{U} - (
        \Tr_{\rA'}\otimes \mathcal{I}_{\rB'})\
        \mathcal{V}\right\|_{\diamond} \\
  &+ \left\|( \Tr_{\rA'}\otimes \mathcal{I}_{\rB'})\ \mathcal{V} -  \Tr_{\rA}\otimes \mathcal{C} \right\|_{\diamond}  \bigr] \\
 =&  \left\| ( \Tr_{\rA'}\otimes \mathcal{I}_{\rB'})\ \mathcal{U} - (
    \Tr_{\rA'}\otimes \mathcal{I}_{\rB'})\
    \mathcal{V}\right\|_{\diamond} \\ &+\min_{\mathcal{C}\in \mathfrak{C}(\rB,\rB')} \left\|( \Tr_{\rA'}\otimes \mathcal{I}_{\rB'})\ \mathcal{V} -  \Tr_{\rA}\otimes \mathcal{C} \right\|_{\diamond}  
  \\
  \leq &  \|\mathcal{U}-\mathcal{V}\|_{\diamond} + \Sigma(\mathcal{V}).
\end{align*}

In the last inequality we used the fact that the trace norm is contractive for the partial trace~\cite{nielsen_chuang_2010}. Therefore $\Sigma(\mathcal{U})-\Sigma(\mathcal{V}) \leq \|\mathcal{U}-\mathcal{V}\|_{\diamond}$. Since the roles of $\mathcal{U}$ and $\mathcal{V}$ are symmetrical, we have $|\Sigma(\mathcal{U})-\Sigma(\mathcal{V})| \leq \|\mathcal{U}-\mathcal{V}\|_{\diamond}$. Therefore $\Sigma$ is 1-Lipschitz continuous with respect to the diamond norm (and therefore continuous).

For the causal influence, we can similarly show that for any pair of unitary channels $\mathcal{U}$ and $\mathcal{V}$, $
|C(\mathcal{U})-C(\mathcal{V})|\leq \|
\mathcal{T}(\mathcal{U})-\mathcal{T}(\mathcal{V})\|_{\diamond}$. Therefore,
it suffices to prove that $\mathcal{T}$ is continuous. First, the map
$A \mapsto A^{-1}$ over $\mathbb{GL}_n(\mathbb{C})$ is continuous for all
$n$. 
Let
$\mathcal{U}_0\in \mathfrak{U}(\rA\rB,\rA'\rB')$ and $\varepsilon
>0$. For all unitary channel $\mathcal{V}$ one has
\begin{align*}
  &\ \|\mathcal{T}(\mathcal{U}_0)-\mathcal{T}(\mathcal{V})\|_{\diamond}\\
  =&\ \|(\mathcal{I}_\rA\otimes\mathcal{U}_0) \circ
     (\mathcal{S}\otimes \mathcal{I}_\rB) \circ (\mathcal{I}_\rA
     \otimes\mathcal{U}^{-1}_0)  \\
  &- (\mathcal{I}_\rA\otimes\mathcal{V}) \circ (\mathcal{S}\otimes \mathcal{I}_\rB) \circ (\mathcal{I}_\rA \otimes\mathcal{V}^{-1})\|_{\diamond} \\
  \leq & \ \|(\mathcal{I}_\rA\otimes\mathcal{U}_0) \circ
         (\mathcal{S}\otimes \mathcal{I}_\rB) \circ (\mathcal{I}_\rA
         \otimes\mathcal{U}^{-1}_0) \\
  &- (\mathcal{I}_\rA\otimes\mathcal{U}_0) \circ (\mathcal{S}\otimes \mathcal{I}_\rB) \circ (\mathcal{I}_\rA \otimes\mathcal{V}^{-1})\|_{\diamond}  \\
  &+ \|(\mathcal{I}_\rA\otimes\mathcal{U}_0) \circ
    (\mathcal{S}\otimes \mathcal{I}_\rB) \circ (\mathcal{I}_\rA
    \otimes\mathcal{V}^{-1}) \\
  & - (\mathcal{I}_A\otimes\mathcal{V}) \circ (\mathcal{S}\otimes \mathcal{I}_\rB) \circ (\mathcal{I}_\rA \otimes\mathcal{V}^{-1})\|_{\diamond} \\
  \leq & \ \|(\mathcal{I}_\rA\otimes\mathcal{U}_0) \circ( \mathcal{S}\otimes \mathcal{I}_\rB)\|_{\diamond} \ \| \mathcal{I}_\rA \otimes\mathcal{U}^{-1}_0 - \mathcal{I}_\rA \otimes\mathcal{V}^{-1}\|_{\diamond}  \\
  & + \|\mathcal{I}_\rA\otimes\mathcal{U}_0 - \mathcal{I}_\rA\otimes\mathcal{V}\|_{\diamond} \ \| (\mathcal{S}\otimes \mathcal{I}_\rB) \circ (\mathcal{I}_A \otimes\mathcal{V}^{-1})\|_{\diamond} \\
  \leq & \ \| \mathcal{I}_\rA \otimes\mathcal{U}^{-1}_0 - \mathcal{I}_\rA \otimes\mathcal{V}^{-1}\|_{\diamond} +  \|\mathcal{I}_\rA\otimes\mathcal{U}_0 - \mathcal{I}_\rA\otimes\mathcal{V}\|_{\diamond}\\  
  =& \ \|\mathcal{U}_0^{-1} - \mathcal{V}^{-1}\|_{\diamond} + \|\mathcal{U}_0 - \mathcal{V}\|_{\diamond},
\end{align*}
where in the third line we use the fact that the diamond norm is
multiplicative with respect to tensor product, in the fourth line we
use the fact that the diamond norm of quantum channels is 1, and in
the last line we use the stability of $\|\cdot\|_{\diamond}$ under
tensor with the identity (see Proposition~\ref{prop:stab_id}). 
Therefore, by choosing $\mathcal{V}$
such that
$\|\mathcal{U}_0^{-1} - \mathcal{V}^{-1}\|_{\diamond}\leq
\frac{\varepsilon}{2}$ and
$\|\mathcal{U}_0 - \mathcal{V}\|_{\diamond} \leq
\frac{\varepsilon}{2}$ (which is possible since the inversion map is
continuous), we have
$
\|\mathcal{T}(\mathcal{U}_0)-\mathcal{T}(\mathcal{V})\|_{\diamond}\leq
\varepsilon$, which concludes the proof.
\end{proof}

\subsubsection{Monotonicity for tensor products of unitary channels
  and asymptotic quantities}

Another basic question on functions $\Sigma$,
$C$ of~Eqs.~\eqref{eq:spelled-defs} is if they are monotone with respect
to the tensor product and sequential composition. In this section we
prove that signalling and causal influence are both monotonically
increasing with respect to tensor product of unitaries. On the
other hand they do not exhibit monotonicity for
sequential composition, as observed at the end of next Section in Remark~\ref{rem:monotonicity-sequential}.

Let $\mathcal{U}$ and $\mathcal{V}$ be two bipartite unitary
channels. Let $\rA_1,\rB_1$ denote Alice and Bob's respective input
systems for $\mathcal{U}$ and $\rA_1',\rB_1'$ their output systems. Let
$\rA_2,\rB_2,\rA_2',\rB_2'$ similarly denote Alice and Bob's input and output
systems for $\mathcal{V}$. When considering the quantities
$\Sigma(\mathcal{U}\otimes\mathcal{V})$ and
$C(\mathcal{U}\otimes\mathcal{V})$, we assume that Alice has input
system $\rA\coloneqq \rA_1\rA_2$ and output system $\rA'\coloneqq\rA_1'\rA_2'$, while Bob has
input system $\rB\coloneqq\rB_1\rB_2$ and output system $\rB'\coloneqq\rB_1'\rB_2'$. Notice that
$\Sigma(\mathcal{U}\otimes\mathcal{V})=
\Sigma(\mathcal{V}\otimes\mathcal{U})$ and
$C(\mathcal{U}\otimes\mathcal{V})=C(\mathcal{V}\otimes\mathcal{U})$. We
can then prove monotonicity as follows:

\begin{lemma}[monotonicity]
\label{lem:tensor}
 Given two unitary channels $\mathcal{U}$ and $\mathcal{V}$ one has $\Sigma(\mathcal{U}\otimes\mathcal{V})\geq \Sigma(\mathcal{U})$ and $C(\mathcal{U}\otimes\mathcal{V})\geq C(\mathcal{U})$.
\end{lemma}

\begin{proof}
 We detail the proof in the case of the signalling (the proof for the causal influence is similar).
\onecolumn
 \begin{align*}
  \Sigma(\mathcal{U}\otimes\mathcal{V})&=\min_{\mathcal{C}\in\mathfrak{C}(\rB,\rB')} \left\| \begin{tikzpicture}
	\begin{pgfonlayer}{nodelayer}
		\node [style=none] (0) at (-2, -0.75) {};
		\node [style=none] (1) at (-2, -2.25) {};
		\node [style=none] (9) at (-0.75, -0.75) {};
		\node [style=none] (10) at (-0.75, -2.25) {};
		\node [style=none] (12) at (0.75, -0.75) {};
		\node [style=none] (13) at (2, -0.75) {};
		\node [style=none] (14) at (0.75, -2.25) {};
		\node [style=none] (15) at (2.75, -2.25) {};
		\node [style=medium box] (16) at (0, -1.5) {$\mathcal{V}$};
		\node [style=upground] (17) at (2.2, -0.75) {};
		\node [style=texthere] (19) at (1.375, -0.35) {$\rA_2'$};
		\node [style=texthere] (20) at (2.625, -1.85) {$\rB_2'$};
		\node [style=texthere] (21) at (-1.75, -0.425) {$\rA_2$};
		\node [style=texthere] (22) at (-1.75, -1.925) {$\rB_2$};
		\node [style=none] (23) at (-2, 2.25) {};
		\node [style=none] (24) at (-2, 0.75) {};
		\node [style=none] (25) at (-0.75, 2.25) {};
		\node [style=none] (26) at (-0.75, 0.75) {};
		\node [style=none] (27) at (0.75, 2.25) {};
		\node [style=none] (28) at (2, 2.25) {};
		\node [style=none] (29) at (0.75, 0.75) {};
		\node [style=none] (30) at (2.75, 0.75) {};
		\node [style=medium box] (31) at (0, 1.5) {$\mathcal{U}$};
		\node [style=upground] (32) at (2.2, 2.25) {};
		\node [style=texthere] (33) at (1.375, 2.65) {$\rA'_1$};
		\node [style=texthere] (34) at (2.625, 1.15) {$\rB'_1$};
		\node [style=texthere] (35) at (-1.75, 2.575) {$\rA_1$};
		\node [style=texthere] (36) at (-1.75, 1.075) {$\rB_1$};
	\end{pgfonlayer}
	\begin{pgfonlayer}{edgelayer}
		\draw (0.center) to (9.center);
		\draw (1.center) to (10.center);
		\draw (12.center) to (13.center);
		\draw (14.center) to (15.center);
		\draw (23.center) to (25.center);
		\draw (24.center) to (26.center);
		\draw (27.center) to (28.center);
		\draw (29.center) to (30.center);
	\end{pgfonlayer}
\end{tikzpicture} \ -\ \begin{tikzpicture}
	\begin{pgfonlayer}{nodelayer}
		\node [style=none] (0) at (-1.25, -0.75) {};
		\node [style=none] (1) at (-1.25, -2.25) {};
		\node [style=none] (4) at (-1.25, -2.5) {};
		\node [style=none] (8) at (0, -0.75) {};
		\node [style=none] (9) at (1, -2.25) {};
		\node [style=none] (13) at (2.5, -2.25) {};
		\node [style=none] (14) at (3.625, -2.25) {};
		\node [style=upground] (16) at (0.2, -0.75) {};
		\node [style=texthere] (19) at (-1, -0.425) {$\rA_2$};
		\node [style=texthere] (20) at (-1, -1.925) {$\rB_2$};
		\node [style=texthere] (21) at (3.375, -1.85) {$\rB_2'$};
		\node [style=none] (23) at (-1.25, 2.25) {};
		\node [style=none] (24) at (-1.25, 0.75) {};
		\node [style=none] (25) at (-1.25, 0.5) {};
		\node [style=none] (26) at (0, 2.25) {};
		\node [style=none] (27) at (1, 0.75) {};
		\node [style=none] (28) at (2.5, 0.75) {};
		\node [style=none] (29) at (3.625, 0.75) {};
		\node [style=upground] (30) at (0.2, 2.25) {};
		\node [style=texthere] (31) at (-1, 2.575) {$\rA_1$};
		\node [style=texthere] (32) at (-1, 1.075) {$\rB_1$};
		\node [style=texthere] (33) at (3.375, 1.15) {$\rB_1'$};
		\node [style=none] (34) at (1, 1.25) {};
		\node [style=none] (35) at (2.5, 1.25) {};
		\node [style=none] (36) at (2.5, -2.75) {};
		\node [style=none] (37) at (1, -2.75) {};
		\node [style=texthere] (38) at (1.75, -0.75) {$\mathcal{C}$};
	\end{pgfonlayer}
	\begin{pgfonlayer}{edgelayer}
		\draw (0.center) to (8.center);
		\draw (1.center) to (9.center);
		\draw (13.center) to (14.center);
		\draw (23.center) to (26.center);
		\draw (24.center) to (27.center);
		\draw (28.center) to (29.center);
		\draw (34.center) to (35.center);
		\draw (35.center) to (36.center);
		\draw (36.center) to (37.center);
		\draw (34.center) to (37.center);
	\end{pgfonlayer}
\end{tikzpicture}\right\|_{\diamond}\\
                                        &\geq
                                          \min_{\mathcal{C}\in\mathfrak{C}(\rB,\rB')}
                                          \left\| \begin{tikzpicture}
	\begin{pgfonlayer}{nodelayer}
		\node [style=none] (0) at (-2, -0.75) {};
		\node [style=none] (1) at (-2, -2.25) {};
		\node [style=none] (9) at (-0.75, -0.75) {};
		\node [style=none] (10) at (-0.75, -2.25) {};
		\node [style=none] (12) at (0.75, -0.75) {};
		\node [style=none] (13) at (2, -0.75) {};
		\node [style=none] (14) at (0.75, -2.25) {};
		\node [style=none] (15) at (2, -2.25) {};
		\node [style=medium box] (16) at (0, -1.5) {$\mathcal{V}$};
		\node [style=upground] (17) at (2.2, -0.75) {};
		\node [style=texthere] (19) at (1.375, -0.35) {$\rA_2'$};
		\node [style=texthere] (20) at (1.375, -1.85) {$\rB_2'$};
		\node [style=texthere] (21) at (-1.375, -0.425) {$\rA_2$};
		\node [style=texthere] (22) at (-1.375, -1.925) {$\rB_2$};
		\node [style=none] (23) at (-3, 2.25) {};
		\node [style=none] (24) at (-3, 0.75) {};
		\node [style=none] (25) at (-0.75, 2.25) {};
		\node [style=none] (26) at (-0.75, 0.75) {};
		\node [style=none] (27) at (0.75, 2.25) {};
		\node [style=none] (28) at (2, 2.25) {};
		\node [style=none] (29) at (0.75, 0.75) {};
		\node [style=none] (30) at (2.75, 0.75) {};
		\node [style=medium box] (31) at (0, 1.5) {$\mathcal{U}$};
		\node [style=upground] (32) at (2.2, 2.25) {};
		\node [style=texthere] (33) at (1.375, 2.65) {$\rA'_1$};
		\node [style=texthere] (34) at (2.625, 1.15) {$\rB'_1$};
		\node [style=texthere] (35) at (-2.75, 2.575) {$\rA_1$};
		\node [style=texthere] (36) at (-2.75, 1.075) {$\rB_1$};
		\node [style=none] (38) at (-2, -2.25) {};
		\node [style=none] (39) at (-2, -0.75) {};
		\node [style=none] (40) at (-2, -0.5) {};
		\node [style=none] (41) at (-2, -2.5) {};
		\node [style=none] (42) at (-2.5, -2.5) {};
		\node [style=none] (43) at (-2.5, -0.5) {};
		\node [style=texthere] (44) at (-2.45, -1.5) {$\sigma$};
		\node [style=none] (45) at (2, -2.25) {};
		\node [style=upground] (46) at (2.2, -2.25) {};
	\end{pgfonlayer}
	\begin{pgfonlayer}{edgelayer}
		\draw (0.center) to (9.center);
		\draw (1.center) to (10.center);
		\draw (12.center) to (13.center);
		\draw (14.center) to (15.center);
		\draw (23.center) to (25.center);
		\draw (24.center) to (26.center);
		\draw (27.center) to (28.center);
		\draw (29.center) to (30.center);
		\draw (43.center) to (40.center);
		\draw (41.center) to (40.center);
		\draw [bend right=90] (43.center) to (42.center);
		\draw (42.center) to (41.center);
	\end{pgfonlayer}
\end{tikzpicture}-\begin{tikzpicture}
	\begin{pgfonlayer}{nodelayer}
		\node [style=none] (0) at (-1.25, -0.75) {};
		\node [style=none] (1) at (-1.25, -2.25) {};
		\node [style=none] (4) at (-1, -2.5) {};
		\node [style=none] (8) at (0, -0.75) {};
		\node [style=none] (9) at (1, -2.25) {};
		\node [style=none] (13) at (2.5, -2.25) {};
		\node [style=none] (14) at (3.75, -2.25) {};
		\node [style=upground] (16) at (0.2, -0.75) {};
		\node [style=texthere] (19) at (-0.625, -0.425) {$\rA_2$};
		\node [style=texthere] (20) at (-0.625, -1.925) {$\rB_2$};
		\node [style=texthere] (21) at (3.125, -1.85) {$\rB_2'$};
		\node [style=none] (23) at (-2.25, 2.25) {};
		\node [style=none] (24) at (-2.25, 0.75) {};
		\node [style=none] (25) at (-2.25, 0.5) {};
		\node [style=none] (26) at (0, 2.25) {};
		\node [style=none] (27) at (1, 0.75) {};
		\node [style=none] (28) at (2.5, 0.75) {};
		\node [style=none] (29) at (4.125, 0.75) {};
		\node [style=upground] (30) at (0.2, 2.25) {};
		\node [style=texthere] (31) at (-2, 2.575) {$\rA_1$};
		\node [style=texthere] (32) at (-2, 1.075) {$\rB_1$};
		\node [style=texthere] (33) at (3.875, 1.15) {$\rB_1'$};
		\node [style=none] (34) at (1, 1.25) {};
		\node [style=none] (35) at (2.5, 1.25) {};
		\node [style=none] (36) at (2.5, -2.75) {};
		\node [style=none] (37) at (1, -2.75) {};
		\node [style=texthere] (38) at (1.75, -0.75) {$\mathcal{C}$};
		\node [style=none] (39) at (-1.25, -0.75) {};
		\node [style=none] (40) at (-1.25, -2.25) {};
		\node [style=none] (41) at (-1.25, -2.25) {};
		\node [style=none] (43) at (-1.25, -0.5) {};
		\node [style=none] (44) at (-1.25, -2.5) {};
		\node [style=none] (45) at (-1.75, -2.5) {};
		\node [style=none] (46) at (-1.75, -0.5) {};
		\node [style=texthere] (47) at (-1.7, -1.5) {$\sigma$};
		\node [style=none] (48) at (3.75, -2.25) {};
		\node [style=upground] (49) at (3.95, -2.25) {};
	\end{pgfonlayer}
	\begin{pgfonlayer}{edgelayer}
		\draw (0.center) to (8.center);
		\draw (1.center) to (9.center);
		\draw (13.center) to (14.center);
		\draw (23.center) to (26.center);
		\draw (24.center) to (27.center);
		\draw (28.center) to (29.center);
		\draw (34.center) to (35.center);
		\draw (35.center) to (36.center);
		\draw (36.center) to (37.center);
		\draw (34.center) to (37.center);
		\draw (46.center) to (43.center);
		\draw (44.center) to (43.center);
		\draw [bend right=90] (46.center) to (45.center);
		\draw (45.center) to (44.center);
	\end{pgfonlayer}
\end{tikzpicture}\right\|_{\diamond}=\Sigma(\mathcal{U}),
 \end{align*}
 \twocolumn
 where we used the monotonicity of the trace norm with respect to the
 partial trace and with respect to the preparation of an arbitrary
 state $\mathcal{\sigma}\in\st(\rA_2\rB_2)$.
\end{proof}

 Notice that, by definition, $\Sigma(\tU\otimes\tI)\leq\Sigma(\tU)$
 and similarly $C(\tU\otimes\tI)\leq C(\tU)$, thus the above lemma implies
 that both signalling and causal influence are stable under
 tensorisation with the identity. We can  define the asymptotic limit of the signalling and causal influence of tensor powers of unitary channels.
\begin{equation}\label{eq:asymptotic}
 \begin{aligned}
 \Sigma_{\infty}(\mathcal{U})&\coloneqq \lim_{n\rightarrow +\infty}\Sigma(\mathcal{U}^{\otimes n}), \\
 C_{\infty}(\mathcal{U})&\coloneqq \lim_{n\rightarrow +\infty}C(\mathcal{U}^{\otimes n}).
\end{aligned}
\end{equation}
Notice that, as a consequence of Lemma~\ref{lem:tensor},
$\left(\Sigma(\mathcal{U}^{\otimes n})\right)_{n\in\mathbb{N}}$ and
$\left(C(\mathcal{U}^{\otimes n})\right)_{n\in\mathbb{N}}$ are
non-decreasing sequences in $[0,2]$, which ensures that their limits
are well defined.

\section{Examples}
\label{sec:examples}

In this section, we compute signalling and
causal influence for two prototypical channels, (i) the SWAP and the
(ii) controlled NOT (or CNOT) unitary gates. We also study the
signalling and causal influence of their tensor powers, which allows
us to determine the corresponding asymptotic quantities, as defined in
the previous section.

\subsection{The SWAP gate}

The first example we will study is the SWAP gate. The SWAP gate $S$ is defined as the unitary map over $\mathcal{H}_\rA\otimes\mathcal{H}_\rB$ that inverts the order of tensor factors:
\begin{equation*}
  S(\ket{\psi}\otimes\ket{\varphi})= \ket{\varphi}\otimes\ket{\psi}
  \quad \text{for all
  }\ket{\psi}\in\mathcal{H}_\rA,\ket{\varphi}\in\mathcal{H}_\rB.
\end{equation*}
The corresponding quantum channel is written as
$\mathcal{S}\coloneqq[\rho\mapsto S \rho S] $. We study signalling and
causal influence of the SWAP channel. We will always assume that
Alice's input and output systems are $\rA$ and $\rB$ respectively, and
Bob's $\rB$ and $\rA$ (generally, the dimensions of systems $\rA$ and $\rB$ are different): 
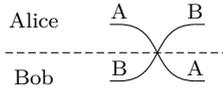
\ctikzfig{swap2}

We begin with causal influence, whose evaluation is by far easier than that of signalling, and we show that $C(\mathcal{S})$ is maximal:

\begin{theorem}[Causal influence of the SWAP for arbitrary dimension]
\label{thm:causinfl_swap}
Let $\tS$ denote the SWAP between two systems $\rA$ and $\rB$ of any arbitrary finite dimension, then $C(\mathcal{S})=2$.
\end{theorem}
\begin{proof}

We have:
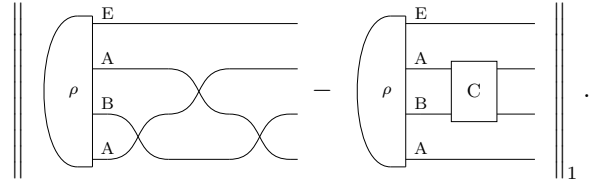
\begin{align*}
&C(\mathcal{S})=\min_{\mathcal{C}\in \mathfrak{C}(\rA\rB)} \max_{\rho
  \in\st(\rE\rA\rB\rA)}\\
  &\left\| \scalebox{0.8}{\begin{tikzpicture}
	\begin{pgfonlayer}{nodelayer}
		\node [style=none] (0) at (-3, 0.75) {};
		\node [style=none] (1) at (-3, -0.75) {};
		\node [style=none] (2) at (-3, -2.25) {};
		\node [style=none] (3) at (-3, 2.25) {};
		\node [style=none] (4) at (-3, -0.75) {};
		\node [style=none] (5) at (-3, -2.25) {};
		\node [style=none] (6) at (-3, 0.75) {};
		\node [style=none] (7) at (-3, 2.5) {};
		\node [style=none] (8) at (-3, -2.5) {};
		\node [style=none] (9) at (-3.5, -2.5) {};
		\node [style=none] (10) at (-3.5, 2.5) {};
		\node [style=texthere] (11) at (-3.625, 0) {$\rho$};
		\node [style=none] (12) at (-2.5, -0.75) {};
		\node [style=none] (13) at (-2.5, -2.25) {};
		\node [style=none] (14) at (-0.5, -2.25) {};
		\node [style=none] (15) at (-0.5, -0.75) {};
		\node [style=none] (16) at (-0.5, 0.75) {};
		\node [style=none] (17) at (1.5, 0.75) {};
		\node [style=none] (18) at (1.5, -0.75) {};
		\node [style=none] (19) at (1.5, -2.25) {};
		\node [style=none] (20) at (3.5, -0.75) {};
		\node [style=none] (21) at (3.5, -2.25) {};
		\node [style=none] (22) at (3.5, 0.75) {};
		\node [style=none] (23) at (3.75, 0.75) {};
		\node [style=none] (24) at (3.75, -0.75) {};
		\node [style=none] (25) at (3.75, -2.25) {};
		\node [style=none] (26) at (3.75, 2.25) {};
		\node [style=texthere] (27) at (-2.5, 2.5) {$\rE$};
		\node [style=texthere] (28) at (-2.5, 1) {$\rA$};
		\node [style=texthere] (29) at (-2.5, -0.5) {$\rB$};
		\node [style=texthere] (30) at (-2.5, -2) {$\rA$};
	\end{pgfonlayer}
	\begin{pgfonlayer}{edgelayer}
		\draw (10.center) to (7.center);
		\draw (8.center) to (7.center);
		\draw [in=180, out=-180, looseness=0.75] (10.center) to (9.center);
		\draw (9.center) to (8.center);
		\draw (3.center) to (26.center);
		\draw (6.center) to (16.center);
		\draw [in=180, out=0, looseness=1.25] (16.center) to (18.center);
		\draw [in=180, out=0, looseness=1.25] (18.center) to (21.center);
		\draw (21.center) to (25.center);
		\draw (4.center) to (12.center);
		\draw [in=-180, out=0, looseness=1.25] (12.center) to (14.center);
		\draw (14.center) to (19.center);
		\draw [in=-180, out=0, looseness=1.25] (19.center) to (20.center);
		\draw (20.center) to (24.center);
		\draw (5.center) to (13.center);
		\draw [in=-180, out=0, looseness=1.25] (13.center) to (15.center);
		\draw [in=-180, out=0, looseness=1.25] (15.center) to (17.center);
		\draw (17.center) to (22.center);
		\draw (22.center) to (23.center);
	\end{pgfonlayer}
\end{tikzpicture}} - \scalebox{0.8}{\begin{tikzpicture}
	\begin{pgfonlayer}{nodelayer}
		\node [style=none] (0) at (-1.5, 0.75) {};
		\node [style=none] (1) at (-1.5, -0.75) {};
		\node [style=none] (2) at (-1.5, -2.25) {};
		\node [style=none] (3) at (-1.5, 2.25) {};
		\node [style=none] (4) at (-1.5, -0.75) {};
		\node [style=none] (5) at (-1.5, -2.25) {};
		\node [style=none] (6) at (-1.5, 0.75) {};
		\node [style=none] (7) at (-1.5, 2.5) {};
		\node [style=none] (8) at (-1.5, -2.5) {};
		\node [style=none] (9) at (-2, -2.5) {};
		\node [style=none] (10) at (-2, 2.5) {};
		\node [style=texthere] (11) at (-2.125, 0) {$\rho$};
		\node [style=none] (12) at (-1, -0.75) {};
		\node [style=none] (13) at (2.75, -2.25) {};
		\node [style=none] (15) at (0, -0.75) {};
		\node [style=none] (16) at (0, 0.75) {};
		\node [style=none] (17) at (1.5, 0.75) {};
		\node [style=none] (18) at (1.5, -0.75) {};
		\node [style=none] (20) at (2.75, -0.75) {};
		\node [style=none] (22) at (2.75, 0.75) {};
		\node [style=texthere] (27) at (-1, 2.5) {$\rE$};
		\node [style=texthere] (28) at (-1, 1) {$\rA$};
		\node [style=texthere] (29) at (-1, -0.5) {$\rB$};
		\node [style=texthere] (30) at (-1, -2) {$\rA$};
		\node [style=medium box] (32) at (0.75, 0) {$\rC$};
		\node [style=none] (33) at (2.75, 2.25) {};
	\end{pgfonlayer}
	\begin{pgfonlayer}{edgelayer}
		\draw (10.center) to (7.center);
		\draw (8.center) to (7.center);
		\draw [in=180, out=-180, looseness=0.75] (10.center) to (9.center);
		\draw (9.center) to (8.center);
		\draw (6.center) to (16.center);
		\draw (4.center) to (12.center);
		\draw (5.center) to (13.center);
		\draw (17.center) to (22.center);
		\draw (3.center) to (33.center);
		\draw (12.center) to (15.center);
		\draw (18.center) to (20.center);
	\end{pgfonlayer}
\end{tikzpicture}} \ \right\|_1.
 \end{align*}
 We prove the equality by showing that for every channel
 $\mathcal{C}$, there exists a state $\rho$ such that the above norm
 is equal to 2. 
 Looking at the diagram, it is evident that one can choose $\rho$ as a
 product state whose second and fourth components (both states of
 $\rA$) have orthogonal support.
\end{proof}
 
Consequently, by inequality~\ref{eq:bound_gen}, we know that
$\Sigma(\mathcal{S})\geq \frac{1}{2}$. We now study signalling of the
SWAP gate, and we show that it depends on the dimension $d_\rA$ of
system $\rA$. We first prove that the minimum in the expression of
$\Sigma(\mathcal{S})$ is attained for the depolarizing channel:

\begin{lemma}[Optimal channel for $\Sigma(\mathcal{S})$ and simplified formula]
\label{lem:depolarizing_opt}
The minimum in
$\Sigma(\mathcal{S})=\min_{\mathcal{C}\in\mathfrak{C}(\rB,\rA)}
\Sigma_{\mathcal{C}}(\mathcal{S})$ is attained for the depolarizing
channel $\mathcal{D}:\rho \mapsto \Tr[\rho] I/d_{\rA}$. Moreover, the
expression of $\Sigma(\mathcal{S})$ can be simplified as follows:
 \begin{align}\label{eq:sigSWAP_simpl}
  \Sigma(\mathcal{\mathcal{S}})=\left\| \mathcal{I}_\rA-\mathcal{D} \right\|_\diamond= \max_{\rho \in \st(EA)} \left\|\rho - \rho_E \otimes \frac{I}{d_A} \right\|_1. 
 \end{align}

\end{lemma}

\begin{proof}
Consider the expression of $\Sigma(\mathcal{S})$ in Eq.~\eqref{eq:spelled-defs}, where
$\Sigma_{\mathcal{C}}(\mathcal{S})$ can be rearranged as follows:
\begin{align*}
  \Sigma_{\mathcal{C}}(\mathcal{S}) &=  \left\| ( \Tr_{\rB}\otimes
                                      \mathcal{I}_{\rA})\
                                      \mathcal{S} -
                                      \bigl(
                                      \Tr_{\rA}\otimes \mathcal{C}
                                      \bigl) \right\|_{\diamond}\\
  &=\left\| \mathcal{I}_{\rA} \otimes \Tr_\rB  - \Tr_{\rA}\otimes \mathcal{C} \right\|_{\diamond}.
\end{align*}
Our goal is to show that
$\Sigma_{\mathcal{C}}(\mathcal{S}) \geq
\Sigma_{\mathcal{D}}(\mathcal{S}) $. Let $\mathcal{U}$ be a unitary
channel over system $\rA$. We have:
\begin{align*}
  & \Sigma_{\mathcal{C}}(\mathcal{S})= \max_{\rho \in\st(\rE\rA\rB)} \| [\mathcal{I}_{\rE}\otimes
    \mathcal{U} \otimes \Tr_\rB  \\
  &\qquad \qquad \qquad\qquad - \mathcal{I}_\rE \otimes \Tr_{\rA}\otimes (\mathcal{U} \circ \mathcal{C})](\rho) \|_1 \\
  &= \max_{\rho \in\st(\rE\rA\rB)} \| [\mathcal{I}_{\rE}\otimes
    \mathcal{U} \otimes \Tr_\rB  \\
   &\qquad \qquad \qquad \qquad - \mathcal{I}_\rE \otimes (\Tr_{\rA}\circ \mathcal{U})\otimes (\mathcal{U} \circ \mathcal{C})](\rho) \|_1\\
  &= \max_{\rho \in\st(\rE\rA\rB)} \| [\mathcal{I}_{\rE\rA} \otimes
    \Tr_\rB  - \mathcal{I}_\rE \otimes \Tr_{\rA}\otimes (\mathcal{U}
    \circ \mathcal{C})](\rho) \|_1 \\
  & = \Sigma_{\mathcal{U} \circ \mathcal{C}}(\mathcal{S}),
\end{align*}
where in the first line we used the fact that the trace norm is
preserved under unitary transformations, in the second line we used
the fact that $\mathcal{U}$ is trace-preserving and in the third line
we used the fact that states of the form $(\mathcal{I}_{\rE\rB}\otimes
\mathcal{U})(\rho)$ range over $\st(\rE\rA\rB)$ as $\rho$ ranges over
$\st(\rE\rA\rB)$. Let $\mathbb{U}(d)$ denote the unitary group of degree $d$. Given a unitary matrix $U$, let $\mathcal{U}$ denote its corresponding unitary channel $U(\cdot)U^{\dag}$. We integrate $\Sigma_{\mathcal{C}}(\mathcal{S})$ over the unitary group $\mathbb{U}(d_A)$ with respect to the Haar measure. On the one hand:
\[
 \int_{\mathbb{U}(d_A)}\Sigma_{\mathcal{C}}(\mathcal{S})dU =\Sigma_{\mathcal{C}}(\mathcal{S}).
\]
 On the other hand, using the previous computation we also have:
\begin{align*}
  &\int_{\mathbb{U}(d_\rA)}\Sigma_{\mathcal{C}}(\mathcal{S})dU  \\
  &= \int_{\mathbb{U}(d_\rA)} \max_{\rho \in\st(\rE\rA\rB)} \|
    [\mathcal{I}_{\rE\rA} \otimes \Tr_\rB \\
 & \qquad \qquad\qquad \qquad - \mathcal{I}_\rE \otimes \Tr_{\rA}\otimes (\mathcal{U} \circ \mathcal{C})](\rho) \|_1  dU\\
  &\geq \max_{\rho \in\st(\rE\rA\rB)} \int_{\mathbb{U}(d_\rA)} \|
    [\mathcal{I}_{\rE\rA} \otimes \Tr_\rB \\
  &\qquad \qquad\qquad \qquad - \mathcal{I}_\rE \otimes \Tr_{\rA}\otimes (\mathcal{U} \circ \mathcal{C})](\rho) \|_1  dU \\
  &\geq \max_{\rho \in\st(\rE\rA\rB)} \left\|
    \int_{\mathbb{U}(d_\rA)}  [\mathcal{I}_{\rE\rA} \otimes \Tr_\rB
    \right.\\
& \qquad \qquad\qquad \qquad\left. - \mathcal{I}_\rE \otimes \Tr_{\rA}\otimes (\mathcal{U} \circ \mathcal{C})](\rho)\  dU \right\|_1 \\
  &= \max_{\rho \in\st(\rE\rA\rB)} \| [\mathcal{I}_{\rE\rA} \otimes
    \Tr_\rB ](\rho)\\
  &\qquad \qquad\qquad \qquad -[ \mathcal{I}_\rE \otimes \Tr_{\rA}\otimes (\frac{I}{d_\rA}\circ\Tr_\rB  \circ \mathcal{C})](\rho) \|_1 \\
  &= \max_{\rho \in\st(\rE\rA\rB)} \|   [\mathcal{I}_{\rE\rA}
    \otimes \Tr_\rB ](\rho) \\
  &\qquad \qquad\qquad \qquad- [ \mathcal{I}_\rE \otimes \Tr_{\rA}\otimes (\frac{I}{d_\rA}\circ\Tr_\rB) ](\rho)  \|_1\\
  &=\Sigma_{\mathcal{D}}(\mathcal{S}),
\end{align*}
where we used the well-known identity $\int_{\mathbb{U}(d_A)}\tU=\tD$, and the fact that $\mathcal{C}$ is trace-preserving. In conclusion, for all
channel $\mathcal{C}\in\mathfrak{C}(B,A)$,
$\Sigma_{\mathcal{C}}(\mathcal{S}) \geq
\Sigma_{\mathcal{D}}(\mathcal{S})$, which completes the first part of
the proof.

The expression of $\Sigma(\mathcal{S})$ can now be computed using the
depolarizing channel as follows:
 \begin{align*}
   &  \Sigma(\mathcal{\mathcal{S}})
= \max_{\rho \in\st(\rE\rA\rB)} \left\| \bigl[\mathcal{I}_{\rE\rA}  \otimes \Tr_\rB - \mathcal{I}_\rE \otimes \Tr_{\rA} \otimes \mathcal{D}\bigr] (\rho) \right\|_1 \nonumber \\
   &=\max_{\rho \in\st(\rE\rA\rB)}
     \left\|\scalebox{0.8}{\begin{tikzpicture}
	\begin{pgfonlayer}{nodelayer}
		\node [style=none] (0) at (0, 0.5) {};
		\node [style=none] (1) at (0, -1) {};
		\node [style=none] (2) at (0, 2) {};
		\node [style=none] (3) at (0, 2.25) {};
		\node [style=none] (4) at (0, -1.25) {};
		\node [style=none] (5) at (-0.5, -1.25) {};
		\node [style=none] (6) at (-0.5, 2.25) {};
		\node [style=texthere] (7) at (-0.625, 0.5) {$\rho$};
		\node [style=none] (8) at (1, -1) {};
		\node [style=none] (9) at (2, 0.5) {};
		\node [style=none] (10) at (2, 2) {};
		\node [style=upground] (11) at (1.25, -1) {};
		\node [style=texthere] (12) at (0.5, 2.25) {$E$};
		\node [style=texthere] (13) at (0.5, 0.75) {$A$};
		\node [style=texthere] (14) at (0.5, -0.75) {$B$};
	\end{pgfonlayer}
	\begin{pgfonlayer}{edgelayer}
		\draw (6.center) to (3.center);
		\draw (4.center) to (3.center);
		\draw [bend right=90] (6.center) to (5.center);
		\draw (5.center) to (4.center);
		\draw (2.center) to (10.center);
		\draw (0.center) to (9.center);
		\draw (1.center) to (11);
	\end{pgfonlayer}
\end{tikzpicture}}\ -\
     \scalebox{0.8}{\begin{tikzpicture}
	\begin{pgfonlayer}{nodelayer}
		\node [style=none] (0) at (-0.5, 0.5) {};
		\node [style=none] (1) at (-0.5, -1) {};
		\node [style=none] (2) at (-0.5, 2) {};
		\node [style=none] (3) at (-0.5, 2.25) {};
		\node [style=none] (4) at (-0.5, -1.25) {};
		\node [style=none] (5) at (-1, -1.25) {};
		\node [style=none] (6) at (-1, 2.25) {};
		\node [style=texthere] (7) at (-1.125, 0.5) {$\rho$};
		\node [style=none] (8) at (0.75, -1) {};
		\node [style=none] (9) at (0.825, 0.5) {};
		\node [style=none] (10) at (4.5, 2) {};
		\node [style=upground] (11) at (1, -1) {};
		\node [style=texthere] (12) at (0, 2.25) {$E$};
		\node [style=texthere] (13) at (0, 0.75) {$A$};
		\node [style=texthere] (14) at (0, -0.75) {$B$};
		\node [style=upground] (15) at (1, 0.5) {};
		\node [style=none] (16) at (2.125, -0.55) {};
		\node [style=none] (17) at (3.25, -0.55) {};
		\node [style=none] (18) at (3.25, -1.45) {};
		\node [style=none] (19) at (2.125, -1.45) {};
		\node [style=texthere] (20) at (2.5, -1.075) {$I/d_A$};
		\node [style=none] (21) at (3.25, -1) {};
		\node [style=none] (22) at (4.5, -1) {};
		\node [style=none] (23) at (0.5, -0.25) {};
		\node [style=none] (24) at (0.5, -1.75) {};
		\node [style=none] (25) at (3.625, -1.75) {};
		\node [style=none] (26) at (3.625, -0.25) {};
		\node [style=texthere] (27) at (2, -2.25) {$\mathcal{D}$};
		\node [style=texthere] (28) at (4, -0.75) {$A$};
	\end{pgfonlayer}
	\begin{pgfonlayer}{edgelayer}
		\draw (6.center) to (3.center);
		\draw (4.center) to (3.center);
		\draw [bend right=90] (6.center) to (5.center);
		\draw (5.center) to (4.center);
		\draw (2.center) to (10.center);
		\draw (0.center) to (9.center);
		\draw (1.center) to (11);
		\draw (17.center) to (16.center);
		\draw [bend right=90, looseness=1.75] (16.center) to (19.center);
		\draw (19.center) to (18.center);
		\draw (17.center) to (18.center);
		\draw (21.center) to (22.center);
		\draw [style=dashed line] (23.center) to (26.center);
		\draw [style=dashed line] (26.center) to (25.center);
		\draw [style=dashed line] (25.center) to (24.center);
		\draw [style=dashed line] (23.center) to (24.center);
	\end{pgfonlayer}
\end{tikzpicture}}\right\|_1\\
   &=\max_{\rho \in\st(\rE\rA)} \left\|\scalebox{0.8}{\begin{tikzpicture}
	\begin{pgfonlayer}{nodelayer}
		\node [style=none] (0) at (0, -0.75) {};
		\node [style=none] (1) at (0, 0.75) {};
		\node [style=none] (2) at (0, 1) {};
		\node [style=none] (3) at (0, -1) {};
		\node [style=none] (4) at (-0.5, -1) {};
		\node [style=none] (5) at (-0.5, 1) {};
		\node [style=texthere] (6) at (-0.45, 0) {$\rho$};
		\node [style=none] (7) at (0, -0.75) {};
		\node [style=none] (8) at (0, 0.75) {};
		\node [style=none] (9) at (1.5, -0.75) {};
		\node [style=none] (10) at (1.5, 0.75) {};
		\node [style=texthere] (11) at (0.5, -0.5) {$A$};
		\node [style=texthere] (12) at (0.5, 1) {$E$};
	\end{pgfonlayer}
	\begin{pgfonlayer}{edgelayer}
		\draw (5.center) to (2.center);
		\draw (3.center) to (2.center);
		\draw [bend right=90] (5.center) to (4.center);
		\draw (4.center) to (3.center);
		\draw (8.center) to (10.center);
		\draw (7.center) to (9.center);
	\end{pgfonlayer}
\end{tikzpicture}}\ -\ \scalebox{0.8}{\begin{tikzpicture}
	\begin{pgfonlayer}{nodelayer}
		\node [style=none] (23) at (-1.125, -0.75) {};
		\node [style=none] (24) at (-1.125, 0.75) {};
		\node [style=none] (25) at (-0.05, -0.75) {};
		\node [style=none] (26) at (3.25, 0.75) {};
		\node [style=texthere] (27) at (-0.625, -0.5) {$A$};
		\node [style=none] (28) at (-1.125, -0.75) {};
		\node [style=none] (29) at (-1.125, 0.75) {};
		\node [style=none] (30) at (-1.125, 1) {};
		\node [style=none] (31) at (-1.125, -1) {};
		\node [style=none] (32) at (-1.625, -1) {};
		\node [style=none] (33) at (-1.625, 1) {};
		\node [style=texthere] (34) at (-1.575, 0) {$\rho$};
		\node [style=upground] (35) at (0.125, -0.75) {};
		\node [style=none] (36) at (1.25, -0.3) {};
		\node [style=none] (37) at (2.375, -0.3) {};
		\node [style=none] (38) at (2.375, -1.2) {};
		\node [style=none] (39) at (1.25, -1.2) {};
		\node [style=texthere] (40) at (1.625, -0.825) {$I/d_A$};
		\node [style=none] (41) at (2.375, -0.75) {};
		\node [style=none] (42) at (3.25, -0.75) {};
		\node [style=texthere] (43) at (3, -0.5) {$A$};
		\node [style=texthere] (44) at (-0.625, 1) {$E$};
	\end{pgfonlayer}
	\begin{pgfonlayer}{edgelayer}
		\draw (24.center) to (26.center);
		\draw (23.center) to (25.center);
		\draw (33.center) to (30.center);
		\draw (31.center) to (30.center);
		\draw [bend right=90] (33.center) to (32.center);
		\draw (32.center) to (31.center);
		\draw (37.center) to (36.center);
		\draw [bend right=90, looseness=1.75] (36.center) to (39.center);
		\draw (39.center) to (38.center);
		\draw (37.center) to (38.center);
		\draw (41.center) to (42.center);
	\end{pgfonlayer}
\end{tikzpicture}}\right\|_1,
 \end{align*}
which proves Eq.~\eqref{eq:sigSWAP_simpl}. 
\end{proof}

We notice that, in the two last lines (see Appendix~\ref{app:norms}
for the properties of the diamond norm), $\rE$ can be chosen to be isomorphic to
$\rA$ instead of $\rA\rB$. This additionally shows that
$\Sigma(\mathcal{S})$ does not depend on the dimension of Bob's input
system $d_\rB$. We now calculate $\Sigma(\mathcal{S})$.

\begin{theorem}[Signalling of $\mathcal{S}$ for arbitrary dimension]\label{thm:signalling-swap}
Let $\tS$ denote the SWAP between two systems $\rA$ and $\rB$ of any arbitrary finite dimension, then one has
\begin{align}
\Sigma(\mathcal{S})=2 (d_\rA^2 - 1)/d_\rA^2.
\end{align}
Moreover, the maximum in
Eq.~\eqref{eq:spelled-defs-2} can be achieved via any maximally entangled state.
\end{theorem}

\begin{proof}
  In Lemma~\ref{lem:depolarizing_opt} we proved that $\Sigma(\mathcal{S})$
  corresponds to the diamond norm distance between the identity and
  the depolarizing channel, which can be realized as random unitaries
  from the same set of unitary transformations as follows:
\begin{align*}
\mathcal{D}(\cdot) =\frac{1}{d_\rA^2}\sum_{i=1}^{d_\rA^2} U_i(\cdot) U_i^\dag,
\end{align*}
where $\{U_i\}_{i=1}^{d^2_\rA}$ is an arbitrary orthogonal basis
($\Tr[U_iU_j^\dag]=\delta_{ij}$) of unitary matrixes for the unitary group
$\mathbb{U}(d_A)$ that includes the identity. Therefore, we can apply
  Lemma~\ref{lem:distance-random-unitaries} concluding the proof.
\end{proof}

\begin{remark}[The case of $\mathcal{S}^{\otimes n}$]\emph{ The above
    proposition also allows us to calculate the SWAP asymptotic
    signalling and causal influence~\eqref{eq:asymptotic}. Notice that
    $\mathcal{S}^{\otimes n}$ is simply a SWAP channel where $d_A$ was
    replaced by $d_A^n$ and $d_B$ by $d_B^n$. Since
    $2 (d_A^2 - 1)/d_A^2$ converges to 2 for $d_A\rightarrow \infty$,
    we have $\Sigma_{\infty}(\mathcal{S})=2$. We also have
    $C_{\infty}(\mathcal{S})=2$ because causal influence of the
    SWAP is maximal regardless of the dimension
    (Theorem~\ref{thm:causinfl_swap}).}
\end{remark}

We can now compare the SWAP to other quantum channels. Since the SWAP
exchanges the two systems, one would expect its signalling to be
maximal. Although we have seen that $\Sigma(\mathcal{S})<2$ regardless
of the system's dimension, it is indeed maximal, in the sense that the
signalling of the SWAP gate is greater than the signalling of any
other channel for which Bob's output system $\rB'$ has the same
dimension as $\rA$.  For this reason, its signalling $\Sigma(\tS)$
represents a benchmark value that we can use to compare our results on
other unitary channels. We prove the maximality of the SWAP's
signalling in the following lemma:

\begin{lemma}[The SWAP channel is the ``most signalling'' channel]\label{lem:swap-most-sig}

 Let $\mathcal{U}$ be a bipartite unitary channel with input systems $\rA,\rB$ and output systems $\rA',\rB'$, and $d_{\rA'}=d_\rB$, $d_{\rB'}=d_\rA$. Then
$\Sigma(\mathcal{U})\leq\Sigma(\mathcal{S})$, where the $\mathcal{S}$ is a SWAP channel.

\end{lemma}

\begin{proof}
We consider the upper bound on $\Sigma(\mathcal{U})=\min_{\mathcal{C}\in\mathfrak{C}(B,B')} \Sigma_{\mathcal{C}}(\mathcal{U})$ obtained when $\mathcal{C}$ is the depolarizing channel, which we show is equal to the signalling of the SWAP transformation $\Sigma(\mathcal{S})$.
 \begin{align*}
   &\Sigma(\mathcal{U}) 
   \leq \max_{\rho \in \st(\rE\rA\rB)} \|
     [(\mathcal{I}_\rE \otimes( \Tr_{\rA'}\otimes
     \mathcal{I}_{\rB'})\ \mathcal{U}) \\
& \qquad\qquad\qquad\qquad- (\mathcal{I}_\rE \otimes \Tr_{\rA}\otimes \frac{I}{d_{\rB'}}\circ\Tr_\rB' )](\rho) \|_1 \\
  &= \max_{\rho \in \st(\rE\rA\rB)} \| [(\mathcal{I}_\rE \otimes( \Tr_{\rA'}\otimes \mathcal{I}_{\rB'})\ \mathcal{U}\bigr) \\
&\qquad\qquad\qquad\qquad- (\mathcal{I}_\rE \otimes (\Tr_{\rA'}\otimes \frac{I}{d_{\rB'}} \circ\Tr_{\rB'})\mathcal{U})](\rho)\|_1 \\
  &=\max_{\rho \in \st(\rE\rA'\rB')} \| [(\mathcal{I}_\rE \otimes \Tr_{\rA'}\otimes \mathcal{I}_{\rB'}) \\
  &\qquad\qquad\qquad\qquad- (\mathcal{I}_\rE \otimes \Tr_{\rA'}\otimes \frac{I}{d_{\rB'}}\circ\Tr_{\rB'})](\rho) \|_1 \\
  &=\max_{\rho \in \st(\rE\rB')} \| [\mathcal{I}_{\rE\rB'} -
    (\mathcal{I}_\rE \otimes  \frac{I}{d_{\rB'}}\circ\Tr_{\rB'}
    )](\rho)\|_1\\
   &=\max_{\rho \in \st(\rE\rB')} \| \rho - \rho_\rE
    \otimes \frac{I}{d_{\rB'}} \|_1 = \Sigma(\mathcal{S}).
  \end{align*}

  In the second line, we use the fact that $\mathcal{U}$ is
  trace-preserving while in the third line we use the fact that states
  of the form $(\mathcal{I}_E\otimes \mathcal{U})(\rho)$ range over
  $\st(\rE\rA'\rB')$ as $\rho$ ranges over $\st(\rE\rA\rB)$. In the
  fourth line we may assume that $\rE\cong \rB'$ (see Appendix~\ref{app:norms}
for the properties of the diamond norm).
\end{proof}

Combining the two previous results, we deduce that for any unitary
channel $\mathcal{U}$, we have $\Sigma(\mathcal{U})<2$, differently
from causal influence which can achieve value 2
(e.~g.~$C(\mathcal{S})=2$ as proved in
Theorem~\ref{thm:causinfl_swap}). Nonetheless, we have seen in
Theorem~\ref{thm:signalling-swap} and subsequent remark that
$\Sigma(\mathcal{S})$ can be arbitrarily close to 2 by choosing
systems of high enough dimension.

\begin{remark}[Gap between causal influence and signalling for SWAP]
\emph{We observe that, thanks to Lemma~\ref{lem:swap-most-sig} we conclude
  that both causal influence and signalling are maximal for the SWAP
  gate. In this case we cannot infer an actual difference between
  these two causal relations, even if they numerically differ. In
  other words we have not found a channel with maximal causal
  influence but limited signalling, which would imply unequivocally
  the existence of causal effects that cannot be understood in terms
  of (cannot be due to) the signalling power of the channel. This will be
  instead the case for the CNOT gate studied in the next section.}
\end{remark}

\subsection{The controlled NOT gate}

In this section, we study another example: the 2-qubit CNOT gate defined as follows: 
\[
  \text{CNOT}(\ket{a}\otimes\ket{b})= \ket{a\oplus b}\otimes \ket{b}
  \]
for all computational basis vector $\ket{a}\otimes \ket{b}$, $a,b=0,1$, where $\oplus$ is addition modulo 2, the fist qubit is called the
\emph{target} qubit and the second is called the \emph{control}
qubit. We will assume that Alice's input and output systems correspond
to the first qubit $\mathcal{H}_\rA=\mathcal{H}_{\rA'}=\mathbb{C}^2$, and
Bob's to the second $\mathcal{H}_\rB=\mathcal{H}_{\rB'}=\mathbb{C}^2$.

We write the corresponding quantum channel as
$\mathcal{C}_X\coloneqq[\rho\mapsto \text{CNOT}\ \rho \ \text{CNOT}] $, and diagrammatically
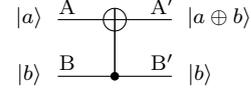
\begin{equation*}
\scalebox{1}{\begin{tikzpicture}
	\begin{pgfonlayer}{nodelayer}
		\node [style=dot] (0) at (0, -0.75) {};
		\node [style=none] (1) at (-1.5, 0.75) {};
		\node [style=none] (2) at (1.5, 0.75) {};
		\node [style=none] (3) at (-1.5, -0.75) {};
		\node [style=none] (4) at (1.5, -0.75) {};
		\node [style=empty dot] (5) at (0, 0.75) {};
		\node [style=none] (6) at (0, 1.05) {};
		\node [style=texthere] (7) at (-2.25, 0.75) {$|a\rangle$};
		\node [style=texthere] (8) at (-2.25, -0.75) {$|b\rangle$};
		\node [style=texthere] (9) at (2.25, -0.75) {$|b\rangle$};
		\node [style=texthere] (10) at (2.75, 0.75) {$|a\oplus b\rangle$};
		\node [style=texthere] (11) at (-1.25, 1) {$\rA$};
		\node [style=texthere] (12) at (1.25, 1) {$\rA'$};
		\node [style=texthere] (13) at (-1.25, -0.5) {$\rB$};
		\node [style=texthere] (14) at (1.25, -0.5) {$\rB'$};
	\end{pgfonlayer}
	\begin{pgfonlayer}{edgelayer}
		\draw (1.center) to (2.center);
		\draw (3.center) to (4.center);
		\draw (0) to (6.center);
		\draw (6.center) to (0);
	\end{pgfonlayer}
\end{tikzpicture}}.
\end{equation*}

We study the signalling and causal influence for arbitrary tensor powers of
the CNOT channel, namely $\mathcal{C}_X^{\otimes n}$
($n\geq 1$) acting on $2n$ qubits. Like previously, each CNOT has
input and output systems $\rA,\rB,\rA',\rB'$, each corresponding to a
single qubit. Alice's overall input (resp. output) system is given by
the tensor product of the input (resp. output) systems of each
individual CNOT, that is $\rA^{\otimes n}$ (resp. $\rA'^{\otimes n}$),
and similarly for Bob.

We begin by calculating the causal influence of the CNOT, which we
show is maximal like in the case of the SWAP gate:

\begin{theorem}[Causal influence of $\mathcal{C}_X^{\otimes n}$]
\label{thm:causinfl_cnot}
Let $\tC$ denote the \emph{CNOT} of qubits, then one has $C(\mathcal{C}_X^{\otimes n})=2$
for all $n\geq 1$.
\end{theorem}
\begin{proof}

We start considering the case $n=1$. We have:
\begin{align*}
&  C(\mathcal{C}_X) 
  = \min_{\mathcal{C}\in \mathfrak{C}(\rA\rA')} \max_{\rho \in
    \st(\rE\rA\rA'\rB')} \\
  &\qquad\qquad\quad\left\| \bigl[(\mathcal{I}_\rE\otimes \mathcal{T}(\mathcal{C}_X)) -  (\mathcal{I}_\rE\otimes \mathcal{C} \otimes \mathcal{I}_{\rB'})\bigr](\rho)\right\|_1 \\
&=\min_{\mathcal{C}\in \mathfrak{C}(\rA\rB)} \max_{\rho \in
                                                                                                                                                                   \st(\rE\rA\rA'\rB')}\\
  &\left\| \scalebox{0.8}{\begin{tikzpicture}
	\begin{pgfonlayer}{nodelayer}
		\node [style=none] (0) at (-3, 0.75) {};
		\node [style=none] (1) at (-3, -0.75) {};
		\node [style=none] (2) at (-3, -2.25) {};
		\node [style=none] (3) at (-3, 2.25) {};
		\node [style=none] (4) at (-3, -0.75) {};
		\node [style=none] (5) at (-3, -2.25) {};
		\node [style=none] (6) at (-3, 0.75) {};
		\node [style=none] (7) at (-3, 2.5) {};
		\node [style=none] (8) at (-3, -2.5) {};
		\node [style=none] (9) at (-3.5, -2.5) {};
		\node [style=none] (10) at (-3.5, 2.5) {};
		\node [style=texthere] (11) at (-3.625, 0) {$\rho$};
		\node [style=none] (12) at (-2.5, -0.75) {};
		\node [style=none] (13) at (-2.5, -2.25) {};
		\node [style=none] (14) at (-0.5, -2.25) {};
		\node [style=none] (15) at (-0.5, -0.75) {};
		\node [style=none] (16) at (-0.5, 0.75) {};
		\node [style=none] (17) at (1.5, 0.75) {};
		\node [style=none] (18) at (1.5, -0.75) {};
		\node [style=none] (19) at (1.5, -2.25) {};
		\node [style=none] (20) at (3.5, -0.75) {};
		\node [style=none] (21) at (3.5, -2.25) {};
		\node [style=none] (22) at (3.5, 0.75) {};
		\node [style=none] (23) at (3.75, 0.75) {};
		\node [style=none] (24) at (3.75, -0.75) {};
		\node [style=none] (25) at (3.75, -2.25) {};
		\node [style=none] (26) at (3.75, 2.25) {};
		\node [style=texthere] (27) at (-2.5, 2.5) {$\rE$};
		\node [style=texthere] (28) at (-2.5, 1) {$\rA$};
		\node [style=texthere] (29) at (-2.5, -0.5) {$\rA'$};
		\node [style=texthere] (30) at (-2.5, -2) {$\rB'$};
		\node [style=dot] (31) at (2.5, -2.25) {};
		\node [style=empty dot] (32) at (2.5, -0.75) {};
		\node [style=none] (33) at (2.5, -0.45) {};
		\node [style=dot] (34) at (-1.5, -2.25) {};
		\node [style=empty dot] (35) at (-1.5, -0.75) {};
		\node [style=none] (36) at (-1.5, -0.45) {};
		\node [style=texthere] (37) at (3.5, 1) {$\rA$};
		\node [style=texthere] (38) at (3.5, -0.5) {$\rA'$};
		\node [style=texthere] (39) at (3.5, -2) {$\rB'$};
		\node [style=texthere] (40) at (-0.75, -2) {$\rB$};
	\end{pgfonlayer}
	\begin{pgfonlayer}{edgelayer}
		\draw (10.center) to (7.center);
		\draw (8.center) to (7.center);
		\draw [in=180, out=-180, looseness=0.75] (10.center) to (9.center);
		\draw (9.center) to (8.center);
		\draw (3.center) to (26.center);
		\draw (6.center) to (16.center);
		\draw [in=180, out=0, looseness=1.25] (16.center) to (18.center);
		\draw (21.center) to (25.center);
		\draw (4.center) to (12.center);
		\draw (14.center) to (19.center);
		\draw (20.center) to (24.center);
		\draw (5.center) to (13.center);
		\draw [in=-180, out=0, looseness=1.25] (15.center) to (17.center);
		\draw (17.center) to (22.center);
		\draw (22.center) to (23.center);
		\draw (31) to (33.center);
		\draw (33.center) to (31);
		\draw (34) to (36.center);
		\draw (36.center) to (34);
		\draw (12.center) to (15.center);
		\draw (13.center) to (14.center);
		\draw (18.center) to (24.center);
		\draw (19.center) to (21.center);
	\end{pgfonlayer}
\end{tikzpicture}}\ - \ \scalebox{0.8}{\begin{tikzpicture}
	\begin{pgfonlayer}{nodelayer}
		\node [style=none] (0) at (-1.5, 0.75) {};
		\node [style=none] (1) at (-1.5, -0.75) {};
		\node [style=none] (2) at (-1.5, -2.25) {};
		\node [style=none] (3) at (-1.5, 2.25) {};
		\node [style=none] (4) at (-1.5, -0.75) {};
		\node [style=none] (5) at (-1.5, -2.25) {};
		\node [style=none] (6) at (-1.5, 0.75) {};
		\node [style=none] (7) at (-1.5, 2.5) {};
		\node [style=none] (8) at (-1.5, -2.5) {};
		\node [style=none] (9) at (-2, -2.5) {};
		\node [style=none] (10) at (-2, 2.5) {};
		\node [style=texthere] (11) at (-2.125, 0) {$\rho$};
		\node [style=none] (12) at (-1, -0.75) {};
		\node [style=none] (13) at (2.75, -2.25) {};
		\node [style=none] (15) at (0, -0.75) {};
		\node [style=none] (16) at (0, 0.75) {};
		\node [style=none] (17) at (1.5, 0.75) {};
		\node [style=none] (18) at (1.5, -0.75) {};
		\node [style=none] (20) at (2.75, -0.75) {};
		\node [style=none] (22) at (2.75, 0.75) {};
		\node [style=texthere] (27) at (-1, 2.5) {$\rE$};
		\node [style=texthere] (28) at (-1, 1) {$\rA$};
		\node [style=texthere] (29) at (-1, -0.5) {$\rA'$};
		\node [style=texthere] (30) at (-1, -2) {$\rB'$};
		\node [style=medium box] (32) at (0.75, 0) {$\rC$};
		\node [style=none] (33) at (2.75, 2.25) {};
		\node [style=texthere] (34) at (2.5, 1) {$\rA$};
		\node [style=texthere] (35) at (2.5, -0.5) {$\rA'$};
	\end{pgfonlayer}
	\begin{pgfonlayer}{edgelayer}
		\draw (10.center) to (7.center);
		\draw (8.center) to (7.center);
		\draw [in=180, out=-180, looseness=0.75] (10.center) to (9.center);
		\draw (9.center) to (8.center);
		\draw (6.center) to (16.center);
		\draw (4.center) to (12.center);
		\draw (5.center) to (13.center);
		\draw (17.center) to (22.center);
		\draw (3.center) to (33.center);
		\draw (12.center) to (15.center);
		\draw (18.center) to (20.center);
	\end{pgfonlayer}
\end{tikzpicture}} \ \right\|_1,
 \end{align*}
 because CNOT is equal to its inverse. We prove the result by showing
 that for any choice of channel $\mathcal{C}$, one can always find a
 state $\rho$ such that
 $\| [(\mathcal{I}_\rE\otimes \mathcal{T}(\mathcal{C}_X)) -
   (\mathcal{I}_\rE\otimes \mathcal{C} \otimes
   \mathcal{I}_{\rB'})](\rho)\|_1=2$. Let
 $\rho = \ket{\psi}\bra{\psi}$ with $\ket{\psi}=\ket{x-++}$, where $x$
 is an arbitrary unit vector of $\mathcal{H}_\rE$, and $\ket{+}$
 and $\ket{-}$ are defined by
 $\ket{\pm}\coloneqq\frac{\ket{0}\pm\ket{1}}{\sqrt{2}}$. We can show
 that on the one hand,
 $(\mathcal{I}_\rE\otimes
 \mathcal{T}(\mathcal{C}_X))(\rho)=\ket{\psi'}\bra{\psi'}$ where
 $\ket{\psi'}=\ket{x+--}$, whereas on the other hand
 $(\mathcal{I}_\rE\otimes \mathcal{C} \otimes
 \mathcal{I}_{\rB'})(\rho)=\ket{x}\bra{x} \otimes
 \mathcal{C}(\ket{-+}\bra{-+})\otimes \ket{+}\bra{+}$.  These two
 states are perfectly distinguishable, by examining the last
 subsystem.

We can now consider the case of arbitrary $n$. By Lemma~\ref{lem:tensor} and
using the fact that $C(\mathcal{C}_X)=2$, one has $C(\mathcal{C}_X^{\otimes n})=2$ for all $n\geq 1$.
\end{proof}

As a corollary of the above lemma one aslo has the asymptotic behaviour $C_{\infty}(\mathcal{C}_X)=2$.

To evaluate signalling, we first exhibit an optimal channel
achieving the infimum (actually a minimum) in its definition: the measurement channel in the computational basis
applied to every output qubit in parallel.

\begin{lemma}[optimal channel for $\Sigma(\mathcal{C}_X^{\otimes n})$
  and simplified formula]\label{lem:optimal-channel-sig-cnot}
  We have
  $\Sigma(\mathcal{C}_X^{\otimes
    n})=\min_{\mathcal{C}\in\mathfrak{C}(B^{\otimes n},B'^{\otimes
      n})} \Sigma_{\mathcal{C}}(\mathcal{C}_X^{\otimes n})$. The
  minimum is attained for the channel $\mathcal{M}^{\otimes n}$, where
  $\mathcal{M} : \rho \mapsto P_0 \rho P_0 + P_1 \rho P_1$ is the
  channel that performs a measurement in the computational
  basis. Moreover, one has the following simplified expression for the
  causal influence
 \begin{equation}\label{eq:sigCn_simpl}
\begin{aligned}
   \Sigma(\mathcal{C}_X^{\otimes n})  &=\left\|
  \mathcal{I}_{\bar{B}}-\mathcal{M}^{\otimes n} \right\|_\diamond\\ & =  \max_{\rho\in \st(E\bar{B})}\left\| \left[\mathcal{I}_{E\bar{B}} - \mathcal{I}_E \otimes \mathcal{M}^{\otimes n}\right](\rho) \right\|_1,
\end{aligned}
\end{equation}
where $\bar{B}=B^{\otimes n}$ and $E\cong \bar{B}$.
\end{lemma}

\begin{proof}

  By Lemma~\ref{lem:infsup} we know that there exists a channel
  $\mathcal{C}\in\mathfrak{C}(B^{\otimes n},B'^{\otimes n})$ such that
  $\Sigma(\mathcal{C}_X^{\otimes n})=
  \Sigma_{\mathcal{C}}(\mathcal{C}_X^{\otimes n})$. Then we prove the following chain
  of implications:
\begin{equation}\label{eq:chain-implications}
  \begin{aligned}
\mathcal{C} \:\text{optimal}\:&\Rightarrow \: \mathcal{M}^{\otimes n} \circ
    \mathcal{C}\circ \mathcal{M}^{\otimes n}\:\text{
    optimal}\\
    &\Rightarrow \:\mathcal{M}^{\otimes n}\:\text{optimal}.
  \end{aligned}
  \end{equation}
  We start by proving the first implication in
  Eq.~\eqref{eq:chain-implications}.  Consider the unitary channel
  $\mathcal{Z}=Z (\cdot)Z$, where $Z=P_0-P_1$, and take the convenient
  notation $\bar{X}\coloneqq X^{\otimes n}$ and
  $\mathcal{Z}_X\coloneqq\mathcal{Z}\otimes \mathcal{I}_{X^{\otimes
      n-1}}$ for all system $X$. Starting from the definition of $\mathcal{C}$ one has:
\begin{align*}
  &\Sigma(\mathcal{C}_X^{\otimes n}) 
  = \max_{\rho \in \st(\rE\bar{\rA}\bar{\rB})} \|
    [(\mathcal{I}_\rE \otimes( \Tr_{\bar{\rA}'}\otimes
    \mathcal{I}_{\bar{\rB'}}) \mathcal{C}_X^{\otimes n} ) \\
  & \qquad \qquad \qquad \qquad \quad- (\mathcal{I}_E \otimes \Tr_{\bar{\rA}}\otimes \mathcal{C})](\rho) \|_1\\
  &= \max_{\rho \in \st(\rE\bar{\rA}\bar{\rB})} \|
    [(\mathcal{I}_\rE \otimes ( (\Tr_{\bar{\rA}'}
    \circ \mathcal{Z}_{\rA'})\otimes
    \mathcal{I}_{\bar{\rB'}})\ \mathcal{C}_X^{\otimes n} )
  \\
  & \qquad \qquad \qquad \qquad \quad -(\mathcal{I}_\rE \otimes \Tr_{\bar{\rA}}\otimes \mathcal{C} )](\rho) \|_1 \\
  &= \max_{\rho \in \st(\rE\bar{\rA}\bar{\rB})} \|
    [(\mathcal{I}_\rE \otimes( (\Tr_{\bar{\rA}'} \circ
    \mathcal{Z}_{\rA'})\otimes \mathcal{I}_{\bar{\rB'}})\
    \mathcal{C}_X^{\otimes n} (\mathcal{Z}_\rA\otimes
    \mathcal{I}_{\bar{\rB}}) ) \\
  &\qquad \qquad \qquad \qquad \quad - (\mathcal{I}_\rE \otimes (\Tr_{\bar{\rA}}\circ \mathcal{Z}_\rA)\otimes \mathcal{C} )](\rho) \|_1 \\
  &=\max_{\rho \in \st(E\bar{\rA}\bar{\rB})} \|
    [(\mathcal{I}_\rE \otimes( \Tr_{\bar{\rA}'} \otimes
    \mathcal{Z}_{\rB'}) \mathcal{C}_X^{\otimes n} \\
  & \qquad \qquad \qquad \qquad \quad - (\mathcal{I}_\rE \otimes \Tr_{\bar{\rA}}\otimes \mathcal{C} )](\rho) \|_1 \\
  &=\max_{\rho \in \st(\rE\bar{\rA}\bar{\rB})} \|
    [(\mathcal{I}_\rE \otimes( \Tr_{\bar{\rA}'} \otimes
    \mathcal{I}_{\bar{\rB}}) \mathcal{C}_X^{\otimes n} \\
  &\qquad \qquad \qquad \qquad \quad - (\mathcal{I}_\rE \otimes \Tr_{\bar{\rA}}\otimes (\mathcal{Z}_{\rB'} \circ \mathcal{C}) )](\rho) \|_1,
\end{align*}
where in the second line, we use the fact that $\mathcal{Z}_\rA$ is
trace-preserving, in the third line we replace $\rho$ with
$\left[\mathcal{I}_\rE\otimes \mathcal{Z}_\rA \otimes
  \mathcal{I}_{\bar{\rB}}\right] (\rho)$, which also ranges over
$\st(\rE\bar{\rA}\bar{\rB})$, and in the fourth line we used the fact
that
$(\mathcal{Z}\otimes\mathcal{I})\mathcal{C}_X
(\mathcal{Z}\otimes\mathcal{I})=
(\mathcal{I}\otimes\mathcal{Z})\mathcal{C}_X$ on the left term and the
fact that $\mathcal{Z}_\rA$ is trace-preserving on the right term. In
the last line, we use the fact that the trace norm is stable under
unitary channels, and that $\mathcal{Z}$ is its own inverse.
Consequently $\mathcal{Z}_{\rB'}\circ \mathcal{C}$ is also an optimal
channel. The set of optimal channels is convex, therefore
$\frac{1}{2}\mathcal{C}+\frac{1}{2} \mathcal{Z}_B\circ \mathcal{C} =
\left(\frac{1}{2}\mathcal{I_{\bar{B'}}} +\frac{1}{2} \mathcal{Z}_{B'}
\right) \circ \mathcal{C}$ is also optimal, and noticing that
$\frac{1}{2}\mathcal{I} +\frac{1}{2} \mathcal{Z} =\mathcal{M}$, one
has that
$ (\mathcal{M}\otimes \mathcal{I}_{B'^{\otimes n-1}})\circ
\mathcal{C}$ is optimal. We can now repeat the same reasoning,
beginning with the optimal channel
$(\mathcal{M}\otimes \mathcal{I}_{B'^{\otimes n-1}})\circ \mathcal{C}$
and using the fact that
$(\mathcal{Z}\otimes\mathcal{I})\mathcal{C}_X
(\mathcal{Z}\otimes\mathcal{I})=
\mathcal{C}_X(\mathcal{I}\otimes\mathcal{Z})$ one proves
that
$(\mathcal{M}\otimes \mathcal{I}_{B'^{\otimes n-1}})\circ
\mathcal{C}\circ (\mathcal{M}\otimes \mathcal{I}_{B'^{\otimes n-1}})$
is an optimal channel. By repeating this process on each input/output
pair, one finally gets that
$\mathcal{M}^{\otimes n} \circ \mathcal{C}\circ \mathcal{M}^{\otimes n}$ is an optimal
channel.

We now prove the second implication in
Eq.~\eqref{eq:chain-implications}. Since $\mathcal{M}^{\otimes n}
\circ \mathcal{C}\circ \mathcal{M}^{\otimes n}$ is optimal, it follows that:
\begin{align*}
&  \Sigma(\mathcal{C}_X^{\otimes n}) \\
  &\quad=  \left\| ( \Tr_{\bar{\rA}'}\otimes \mathcal{I}_{\bar{\rB'}})\
    \mathcal{C}_X^{\otimes n}  - \Tr_{\bar{A}}\otimes
    (\mathcal{M}^{\otimes n } \circ \mathcal{C} \circ
    \mathcal{M}^{\otimes n }) \right\|_{\diamond}\\
  &\quad=  \left\| \Tr_{\bar{\rA}}\otimes \mathcal{I}_{\bar{\rB}}  -
    \bigl(\Tr_{\bar{\rA}'}\otimes (\mathcal{M}^{\otimes n } \circ
    \mathcal{C} \circ \mathcal{M}^{\otimes n })\bigr)\
    \mathcal{C}_X^{\otimes n} \right\|_{\diamond}
  \\
  &\quad=  \left\|  \Tr_{\bar{\rA}}\otimes \mathcal{I}_{\bar{\rB}} -
    \Tr_{\bar{\rA}}\otimes (\mathcal{M}^{\otimes n } \circ \mathcal{C}
    \circ \mathcal{M}^{\otimes n }) \right\|_{\diamond}\\
  &\quad=  \left\| \mathcal{I}_{\bar{\rB}} - \mathcal{M}^{\otimes n } \circ \mathcal{C} \circ \mathcal{M}^{\otimes n } \right\|_{\diamond}\\
  &\quad=  \max_{\rho\in \st(\bar{\rB}\rE)}\left\|\left[ \mathcal{I}_{\bar{\rB}\rE} - \mathcal{M}^{\otimes n } \circ \mathcal{C} \circ \mathcal{M}^{\otimes n }\otimes \mathcal{I}_\rE\right] (\rho)\right\|_1,
 \end{align*}
 where in the second equality, we use the fact that CNOT is its own
 inverse, while in the third equality we used the fact that
 $(\Tr_{\rA}\otimes \mathcal{M})\mathcal{C}_X = \Tr_{\rA}\otimes
 \mathcal{M}$. In the last line, we have $\rE\cong \rB$, and we order
 the systems as $\bar{\rB}\rE$ rather that $\rE\bar{\rB}$ in order to
 simplify the notations in the next part of the proof. We define the subspaces $E_0$ and $E_1$ as
 the image and kernel of the projector
 $ \mathcal{M}^{\otimes n}\otimes \mathcal{I}_\rE$, respectively. Every
 state $\rho\in \st(\bar{\rB}\rE)$ can be decomposed as
 $\rho=\rho_0+\rho_1$ with $\rho_0\in E_0, \rho_1\in E_1$. Therefore:
\begin{align*}
& \Sigma(\mathcal{C}_X^{\otimes n}) \\
 &=\max_{\rho \in \st(\bar{\rB}\rE)} \left\| \rho_0 + \rho_1 -
   \left[\mathcal{M}^{\otimes n } \circ \mathcal{C} \circ
   \mathcal{M}^{\otimes n }\otimes \mathcal{I}_\rE\right](\rho)
   \right\|_1
  \\ & \geq \min_{\substack{N \in E_0: \\ \Tr(N)=0}} \ \max_{\rho \in \st(\bar{\rB}\rE)} \left\| \rho_1 + N \right\|_1,
 \end{align*}
since $\rho_0 - \left[\mathcal{M}^{\otimes n } \circ \mathcal{C} \circ
  \mathcal{M}^{\otimes n} \otimes \mathcal{I}_\rE\right](\rho)$ is in
$E_0$ and has trace 0. By Lemma~\ref{lem:optimalmatrix}, the minimum
is attained when $N=0$ and $ \Sigma(\mathcal{C}_X^{\otimes n}) \geq
\max_{\rho \in \st(\bar{\rB}\rE)} \left\| \rho_1 \right\|_1 =
\max_{\rho \in \st(\bar{\rB}\rE)} \left\| \rho - \left[
    \mathcal{M}^{\otimes n} \otimes
    \mathcal{I}_\rE\right](\rho)\right\|_1$, which shows that
$\mathcal{M}^{\otimes n}$ is optimal, thus
concluding the proof.
\end{proof}

We can now also calculate the signalling for $\mathcal{C}_X^{\otimes n}$.
  \begin{theorem}[Signalling of $\mathcal{C}_X^{\otimes n}$] Let $\tC$ denote the \emph{CNOT} of qubits, then one has
 \begin{align}
 \label{eq:signalling_cnot}
 \begin{split}
   \Sigma(\mathcal{C}_X^{\otimes n}) =2 (d_{\bar{B'}} - 1)/d_{\bar{B'}}
 \end{split}
 \end{align}
with $d_{\bar{B'}}=2^n$. Moreover, the maximum in
Eq.~\eqref{eq:sigCn_simpl} can be achieved via any maximally entangled state.
\end{theorem}
\begin{proof}
  The proof is analogous to that of Theorem~\ref{thm:signalling-swap}
  upon noticing that Eq.~\eqref{eq:sigCn_simpl} corresponds to the
  diamond norm between the channel $\mathcal{M}^{\otimes n}$ and the
  identity channel, and that the former abides by the hypotheses of
  Lemma~\ref{lem:distance-random-unitaries}. Indeed the qubit channel
  $\mathcal{M}$ can be written as a random unitary as
  $\mathcal{M}=\frac{1}{2} \mathcal{I}+ \frac{1}{2}\mathcal{Z}$ where
  $\mathcal{Z}(\cdot)=Z(\cdot)Z$ is the Pauli channel having Kraus
  operator $Z=P_0-P_1$. Therefore, $\mathcal{M}^{\otimes n}$ can be
  realized as random unitaries including the identity as follows
  \begin{align}
    \mathcal{M}^{\otimes
    n}(\cdot)=\frac{1}{2^n}\sum_{j=1}^{2^n} \mathcal{U}_j (\cdot),
\end{align}
with $\{\mathcal{U}_j\}_{j=1}^{2^n}=\{\otimes_{k=1}^{n}
\mathcal{V}_{k},
\mathcal{V}_k\in\{\mathcal{I},\mathcal{Z}\}\}$. We can finally apply
Lemma~\ref{lem:distance-random-unitaries} to conclude the proof.
\end{proof}

The above result is similar to the one obtained for the SWAP gate in
Theorem~\ref{thm:signalling-swap}, namely
$\Sigma(\mathcal{\mathcal{S}}) = 2\frac{d_A^2 - 1}{d_A^2}$ (for the
SWAP gate, Bob's output is system $A$). The fact that the dimension is
not squared here can be explained by the fact that the auxiliary
system $E$ is not actually needed (one can use a factorized state instead of a maximally entangled state to compute the diamond norm).

The asymptotic behaviour of the signalling of CNOT immediately follows from
the last lemma: we have
$2 \frac{d_{\bar{B'}} - 1}{d_{\bar{B'}}} \xrightarrow[n \rightarrow
\infty]{} 2$, therefore
$\Sigma_{\infty}(\mathcal{C}_X^{\otimes n})=2$.

\begin{remark}[Causal influence and signalling are not monotone for
  sequential composition]\label{rem:monotonicity-sequential}
\rm{As we have already observed, contrary to the tensor product operation, $\Sigma(\cU)$ and $C(\cU)$ do not exhibit any monotonicity property with respect to sequential composition. We can provide examples showing that, in general, one can have $\Sigma(\cU\circ\cV)\geq\max\{\Sigma(\cU),\Sigma(\cV)\}$ or $\Sigma(\cU\circ\cV)<\min\{\Sigma(\cU),\Sigma(\cV)\}$ and similarly for the causal influence. Consider two unitaries $\tU,\;\cV\in\mathfrak{U}(\rA\rB)$ and the composition $(\cU\otimes\tI)\circ(\tI\otimes\cV)$ where $\tI$ is the identity on $\rA\rB$. Then, whatever $\tU$ and $\tV$ are:
\begin{align*}
  \Sigma[(\tU\otimes\tI)\circ(\tI\otimes\tV)]&=\Sigma(\tU\otimes\tV)\geq\max\{\Sigma(\tU),\Sigma(\tV)\}\\
                                               &=\max\{\Sigma(\tU\otimes\tI),\Sigma(\tI\otimes\tV)\}.
\end{align*}

As an explicit example take $\tU=\tV={\rm CNOT}$. On the one hand $\max\{\Sigma(\tU\otimes\tI),\Sigma(\tI\otimes\tV)\}=1$, while the tensor product has signaling bounded from below by $2\frac{2^2-1}{2^2}=\frac{3}{2}$, therefore the inequality is strict. In the case of causal influence, the explicit examples at our disposal saturate the inequality, since $C(\tS)=C(\tC_X)=2$. Now consider $\tV=\tU^{-1}$ with $\tU=\tS=\tU^{-1}$. Then

\[
\Sigma(\tU\tU^{-1})=0<\frac{3}{2}=\min\{\Sigma(\tU),\Sigma(\tU^{-1}) \}.
\]

Therefore the inequality is strict, as in the case of the causal influence since $C(\tS)=2$.}
\end{remark}
\section{Outlook and conclusions}
\label{sec:conclusions}

In this work we have pursued the study of two functions, $\Sigma$ and $C$, formerly introduced in~\cite{barsse2024causalinfluenceversussignalling}, which respectively quantify the amount of signalling and causal influence mediated
by a unitary channel. It is worth noticing that the introduced measures,
which establish a hierarchy on the set of unitary channel in terms of
signalling and causal-influencing power, are generally incommensurable
quantities. Therefore, one has to be careful in drawing conclusions
based on the numerical results: if a given channel $\cU$ saturates the
inequality $\Sigma(\cU)\leq C(\cU)$, we cannot conclude that all the
causal influence mediated by $\cU$ has to be ascribed to its
signalling power or, complementarily, if we find $\tU$ such that
$\Sigma(\cU) < C(\cU)$ we cannot claim that its causal influence 
cannot be completely described in terms of signalling.

We have shown that the above two measures are continuous as functions on the set of
unitary channels, in the sense that if two channels are close in
diamond norm, their signalling and causal influence are correspondingly
close. We have also investigated the behaviour of the signalling and
causal influence with respect to parallel and sequential
composition. While it is intuitive that Alice signals and induces more
causal influence on Bob when they share multiple copies of the same
channel, there is no reason to expect such a ``super-activation''
property with respect to sequential composition. Indeed, in the former
case we have proved that for every pair of unitaries $\tU$ and $\tV$
we have $\Sigma(\tU\otimes\tV)\geq\max\{\Sigma(\tU),\Sigma(\tV)\}$,
while for sequential composition we have given counterexamples showing
that one can have
$\Sigma(\tU\circ\tV)\geq\max\{\Sigma(\tU),\Sigma(\tV)\}$ and
$\Sigma(\tU\circ\tV)<\max\{\Sigma(\tU),\Sigma(\tV)\}$ and similarly
for $C$.

\begin{table}[h]
\begin{center}
\begin{tabular}{|c|c|c|c|c|}
\hline\hline
	 & CNOT & SWAP & SWAP $(d)$ &CNOT$^{\otimes n}$ \\[0.3cm]
\hline
$\Sigma(\cU)$ & 1 & $\frac{3}{2}$ & $2\frac{d^2-1}{d^2}$ & $2\frac{2^n-1}{2^n}$\\[0.3cm]
\hline
$C(\cU)$ & 2 & 2 &2 &2 \\[0.3cm]
\hline
\end{tabular}
\end{center}
\caption{}
\label{tab:values}
\end{table}

The general features (continuity and monotonicity) of signalling and
causal influence here derived, have enhanced their analytical
evaluation on two references: the SWAP and the CNOT. These two examples
prove a remnant of the inequivalence between the two causal
relations in quantum theory and allows to conjecture asymptotic
equivalence.

\begin{description}
\item[Gap between signalling and causal influence] (From the first two columns of the table~\ref{tab:values})

  If we take a look at the values of $\Sigma$ and $C$ for the SWAP
  when $\rA$ and $\rB$ are qubits and we compare them with the
  corresponding results of the CNOT we can observe that: i) the
  signalling and causal influence of the SWAP are both maximal, 3/2
  and 2 respectively. This means that all the causal influence may be
  mediated by the channel could be ascribed to the signalling in this
  case, despite the numerical discrepancy ii) The signalling of the
  CNOT, being equal to 1, is not the maximal value that is allowed
  while its causal influence is 2 which is maximal. This proves that
  for the CNOT part of its causal influence is not due to signalling
  but is an extra causal effect. In other words CNOT has the largest
  causal effects possible, while it is less signalling than SWAP. The
  causal effects in communication are then beyond it. The mismatch
  between signalling and causal influence for quantum CNOT
  is reminiscent of the one highlighted by the classical
  CNOT~\cite{Perinotti2021causalinfluencein}.

\item[Evidence of asymptotic equivalence] (From the third and fourth columns of the table~\ref{tab:values})

  Despite the proof that a gap occurs for specific channels $\tU$, e.~g.~the CNOT gate, the prototypical cases of study indicate that the
  gap vanishes in the limit of infinite copies used in parallel, which
  leads to the following conjecture:

  \emph{Conjecture:}
  $ \Sigma_{\infty}(\mathcal{U})\coloneqq \lim_{n\rightarrow
    +\infty}\Sigma(\mathcal{U}^{\otimes n}) =
  C_{\infty}(\mathcal{U})\coloneqq \lim_{n\rightarrow
    +\infty}C(\mathcal{U}^{\otimes n})$ for every $\tU$.

  Indeed, in the case of the SWAP of general dimension and for tensor
  powers of the CNOT, the causal influence has been analytically
  found and is equal to 2 (see second line of the table). 

  We computed the signalling of asymptotic channels, with its value
  going to 2 in the limit of large $n$ (see first line of the
  table). We have identified the optimal channels in the infimum
  (actually a minimum) of Definitions~\ref{def:nos} (see also
  Eq.~\eqref{eq:spelled-defs}): the depolarising channel
  $\tD:\rho\mapsto\frac{I}{d_{\rA}}$ for the SWAP and the tensor power
  of the channel corresponding to a measurement on a single qubit in
  the computational basis $\tM:\rho\mapsto P_0\rho P_0 + p_1\rho P_1$
  for the CNOT$^{\otimes n}$. Then we found the optimal states in the
  supremum (actually a maximum) in Eqs.~\eqref{eq:spelled-defs-2} (see
  also the expressions \eqref{eq:sigSWAP_simpl} and
  \eqref{eq:sigCn_simpl}) for $\Sigma(\tS)$ and
  $\Sigma(\tC_X^{\otimes n})$ respectively. Both for the SWAP in
  arbitrary dimension $d_\rA$ and for tensor powers of the CNOT
  $C_X^{\otimes{n}}$, the maximally entangled state maximises
  Eqs.~\eqref{eq:sigSWAP_simpl} and \eqref{eq:sigCn_simpl}.
   \end{description} 

   The choice of the diamond norm in the definition of $\Sigma$ and
   $C$ is due to its physical connection to the task of channel
   discrimination. While being the most suitable norm if one aims to
   generalize these measures to more general contexts, such as the
   generalised probabilistic theories framework, one could consider
   other norms within quantum theory.  Given that the diamond norm
   can be hard to compute, the search of alternatives (e.g. the
   Hilbert-Schmidt norm, which is easier to evaluate) is of relevance
   in view of application to practical scenarios, e.g. quantum
   cryptography.  However, even if in finite dimensional quantum
   theory all norms are equivalent, one cannot go back and forth from
   one norm to the other, different norms may establish different
   hierarchies, and a careful investigation is needed, both with
   analytical and numerical tools. In this direction, one could
   compare our definitions with quantifiers introduced
   in~\cite{PhysRevLett.129.020501,PhysRevA.106.062415,PhysRevA.108.022222}

\emph{Acknowledgements.} {K. B. acknowledges financial support from the ARPE
  program of ENS Paris-Saclay.  P. P. acknowledges financial support from PNRR MUR
  project PE0000023-NQSTI. A. T. acknowledges the financial support of Elvia
  and Federico Faggin Foundation (Silicon Valley Community Foundation
  Project ID\#2020-214365). L. V. acknowledges financial support from
  the French government under the France 2030 investment plan, as part
of the initiative d'Excellence d'Aix-Marseille Univestit\'e-A*MIDEX, AMX-22-CEI-01.}

\bibliographystyle{unsrt} 
\bibliography{biblio}

\newpage
\begin{appendices}

 \section{Norms}
 \label{app:norms}

 We define several useful norms: first the \emph{trace norm}, which is defined over matrices, and then the \emph{induced trace norm} and \emph{diamond norm}, which are defined over linear superoperators (linear maps that act on matrices). We will also explain why, when it comes to evaluating how close two quantum channels are, the diamond norm is the most relevant notion (see Refs~\cite{10.1145/276698.276708,Watrous_2018} for more details).

 \textbf{The trace norm.} The \emph{trace norm} $\|\cdot\|_1$ is defined over matrices as follows:
 \[
  \|X\|_{1}\coloneqq\Tr\sqrt{X^{\dag}X}
 \]

 If $X$ is Hermitian, $\|X\|_1$ is the sum of the absolute values of its eigenvalues:
 \[
  \|X\|_1=\sum_{\lambda\in Sp(X)} |\lambda|
 \]

The trace norm is closely related to the \emph{trace distance}, which is a standard way of evaluating the closeness of two quantum states. The trace distance between two quantum states $\sigma$ and $\rho$ is given by
\[
 D(\sigma,\rho)\coloneqq\frac{1}{2}\|\sigma-\rho\|_1.
\]
We always have $0\leq D(\sigma,\rho)\leq 1$, with $D(\sigma,\rho)=0$ if and only if $\sigma = \rho$ and $D(\sigma,\rho)=1$ if and only if $\sigma$ and $\rho$ have orthogonal support, i.e. they can be perfectly distinguished by some quantum measurement. In the case of pure states, the trace distance may be reformulated as follows:
\[
 D(\ket{\psi}\bra{\psi},\ket{\varphi}\bra{\varphi})=\sqrt{1-|\bra{\psi}\varphi\rangle|^2}.
\]
Notice that the distance is 1 if and only if the states are orthogonal and 0 if and only if $\ket{\psi}=e^{i\theta}\ket{\varphi}$ for some $\theta$.

\begin{lemma}[Properties of the trace norm]
\label{lem:proptracenorm}
The following properties hold.
\begin{itemize}
 \item For all quantum state $\rho$, $\|\rho\|_1=1$,
 \item $\|X\otimes Y\|_1 = \|X\|_1\|Y\|_1$.
\end{itemize}

\end{lemma}

\textbf{The induced trace norm.} The \emph{induced trace norm} $\vertiii{\cdot}$ is defined as the usual operator norm with respect to the trace norm $\|\cdot\|_1$. For any linear superoperator $\mathcal{C}$:
\[
 \vertiii{\mathcal{C}}\coloneqq\sup_{\rho\neq 0} \frac{\|\mathcal{C}(\rho)\|_1}{\|\rho\|_1} = \sup _{\|\rho\|_1=1} \|\mathcal{C}(\rho)\|_1.
\]
\noindent Here $\rho$ can be any non-zero matrix, not necssarily a quantum state.

This is not the norm we will use in order to evaluate the closeness of quantum channel. The first reason is that the induced trace norm is not stable under tensor product with the identity. Ref~\cite{10.1145/276698.276708} provides the following counterexample: the transpose map $T:\ket{i}\bra{j}\mapsto\ket{j}\bra{i}$ for $i,j=0,1$, acting on a single qubit system.
Moreover, $\vertiii{\cdot}$ isn't multiplicative with respect to tensor products~\cite{Watrous_2018}. 
Therefore, we will use the diamond norm instead, for which both of these properties hold. The definition of the diamond norm is based on that of $\vertiii{\cdot}$:

\textbf{The diamond norm.} The \emph{diamond norm} (or \emph{completely bounded trace norm}) $\|\cdot\|_{\diamond}$ is defined as follows over the space of superoperators:
\[
 \|\mathcal{C}\|_{\diamond}\coloneqq\vertiii{\mathcal{I}_E \otimes \mathcal{C}}
\]
\noindent If $\mathcal{C}$ acts on system $A$, then $E$ is chosen to be isomorphic to $A$.

These two norms have the following properties. Importantly, $\|\cdot\|_{\diamond}$ is muliplicative with respect to tensor products.

\begin{lemma}[Properties of the induced trace norm and diamond norm]
\label{lem:propdiamondkitaev}
The following properties hold.
\begin{itemize}
 \item $\vertiii{\mathcal{D}\circ \mathcal{C}}\leq \vertiii{\mathcal{D}}\vertiii{\mathcal{C}}$,
 \item For all positive and trace preserving map $\mathcal{C}$, $\vertiii{\mathcal{C}}=1$,
 \item $\|\mathcal{D}\circ \mathcal{C}\|_{\diamond} \leq \|\mathcal{D}\|_{\diamond}\|\mathcal{C}\|_{\diamond}$,
 \item $\|\mathcal{D}\otimes \mathcal{C}\|_{\diamond} = \|\mathcal{D}\|_{\diamond}\|\mathcal{C}\|_{\diamond}$,
 \item For all quantum channel $\mathcal{C}$, $\|\mathcal{C}\|_{\diamond}=1$.
\end{itemize}
\end{lemma}
\noindent (See Refs.~\cite{10.1145/276698.276708,Watrous_2018} for proofs of these statements.)

The diamond norm is also stable under tensor with the identity map:
\begin{proposition}
\label{prop:stab_id}
 Let $\mathcal{C}$ be a linear superoperator and $E$ and auxiliary system. We have $\|\mathcal{C}\otimes \mathcal{I}_E\|_{\diamond}=\|\mathcal{C}\|_{\diamond}$.
\end{proposition}

Lastly, it is shown in Ref.~\cite{Watrous_2018} that the diamond norm can be expressed as follows in the case of Hermitian-preserving maps (with our notations):
\begin{lemma}
\label{lem:equalnorms}
Let $\mathcal{C}$ be a Hermitian-preserving linear map acting on system $A$. Then we have:
\[
 \|\mathcal{C}\|_{\diamond} = \sup_{\ket{\psi}\bra{\psi}\in \st(EA)} \|(\mathcal{I}_E\otimes\mathcal{C})(\ket{\psi}\bra{\psi})\|_1,
\]
where $E$ is isomorphic to $A$.
\end{lemma}

This matches the definition we gave in Section~\ref{section:definitions}, where we only considered the norm of Hermitian-preserving maps. (Indeed, we show in Lemma~\ref{lem:infsup} that the supremum can equivalently be taken over pure states.) In particular, the difference of two quantum channels is Hermitian preserving, so we use this expression, which is easier to calculate.

A useful result on the distance between two quantum channels given in
terms of the diamond norm has been proved in
Ref.~\cite{PhysRevA.71.062340}:

\begin{lemma}\label{lem:distance-random-unitaries} Consider two channels
  $\mathcal{E}_1, \mathcal{E}_1\in\mathfrak{C}(\rA)$ that can be
  realized as random unitary transformations from the same set of
  orthogonal unitaries, namely
  $\mathcal{E}_n(\cdot)=\sum_ip^{(n)}_i\mathcal{U}_i(\cdot)$, with
  $\sum_i p_i^{(n)}=1$, for $ n=1,2$. Then one has
\begin{align*}
&  \| \mathcal{E}_1-\mathcal{E}_2 \|_\diamond=\\
&  \max_{\rho\in\st(\rE\rA)}\|
  (\mathcal{I}_\rE\otimes\mathcal{E}_1)\rho-(\mathcal{I}_\rE\otimes\mathcal{E}_2)
                                                              \rho \|_1=2 \sum_i |c_i|,\\
  &c_i=p_i^{(1)}-p_i^{(2)},
\end{align*}
where the maximum is achieved by using an arbitrary maximally entangled
state $\rho$ at the input. 
\end{lemma}

\section{Proofs}

\begin{lemma}
 \label{lem:optimalmatrix}
 Let $\bar{B}$ and $E$ be isomorphic systems each corresponding to $n$ qubits. The quantum channel $\mathcal{M}^{\otimes n} \otimes \mathcal{I}_E$ acts as a projector on the space of Hermitian matrices over $\mathcal{H}_{\bar{B}E}$. Let $E_0$ and $E_1$ denote its image and kernel, respectively. Every state $\rho\in \st(\bar{B}E)$ can thus be decomposed as $\rho=\rho_0+\rho_1$ with $\rho_0\in E_0, \rho_1\in E_1$. Then:
 \begin{align*}
&  \inf_{\substack{N \in E_0: \\ \Tr(N)=0}} \ \max_{\rho \in
    \st(\bar{B}E)} \left\| \rho_1 + N \right\|_1 =\\
&  \min_{\substack{N \in E_0: \\ \Tr(N)=0}} \ \max_{\rho \in \st(\bar{B}E)} \left\| \rho_1 + N \right\|_1 = \max_{\rho \in \st(\bar{B}E)} \left\| \rho_1 \right\|_1.
 \end{align*}

\end{lemma}

\begin{proof}
We have
\begin{align*}
 \mathcal{M}^{\otimes n}(X) =\sum_{b\in\{0,1\}^n} P_b^{(n)}XP_b^{(n)},
\end{align*}
where $P_b^{(n)}\coloneqq \bigotimes_{i=1}^n {P_{b_i}}$. Therefore, $\mathcal{M}^{\otimes n}\otimes \mathcal{I}_E$ acts on Hermitian matrices over $\mathcal{H}_{\bar{B}E}$ by selecting diagonal blocks. $E_0$ is therefore to the space of block-diagonal matrices with respect to the canonical basis $b\in\{0,1\}^n$, and conversely $E_1$ is the space of matrices whose non-zero blocks are off diagonal. The goal is now to prove that the infimum is attained for $N=0$. We must first ensure that the infimum is a minimum. We have $\left\| \rho_1 \right\|_1 \leq \left\| \rho \right\|_1+\left\| \rho_0 \right\|_1=2$, so when $N=0$ we have  $\left\| \rho_1 + N \right\|_1\leq 2$ whereas when $\|N\|_1>4$ we have $\left\| \rho_1 + N \right\|_1\geq |\|\rho_1\|_1 - \|N\|_1|>2$. Therefore it suffices to consider $\|N\|_1\leq 4$. Since the subset of elements of $E_0$ of trace 0 and norm no greater than 4 is a compact space, and $N\mapsto \max_{\rho \in \st(\bar{B}E)} \left\| \rho_1 + N \right\|_1$ is a continuous function, the infimum is a minimum. Moreover, it can be shown that the set of all $N$ for which the minimum is attained is a convex set.
Choose $N\in E_0$ that reaches the infimum. By definition of $E_0$, $N$ has the following block diagonal form:
\[
 N=\begin{bmatrix}
    N_1 &  \cdots & 0 \\
    \vdots & \ddots & \vdots \\
    0 & \cdots & N_{2^n}
   \end{bmatrix}.
\]
The trace norm is stable under application of unitary channels, and $E_0$ and $E_1$ are stable under application of block diagonal unitary channels. Therefore, we may assume without loss of generality that every $N_i$ is a diagonal matrix: 
\[
 N=\begin{bmatrix}
    D_1 &  \cdots & 0 \\
    \vdots & \ddots & \vdots \\
    0 & \cdots & D_{2^n}
   \end{bmatrix}.
\]
Let
\[
 D_1=\begin{bmatrix}
    \lambda_1 &  \cdots & 0 \\
    \vdots & \ddots & \vdots \\
    0 & \cdots & \lambda_{2^n}
   \end{bmatrix},
\]
be the first diagonal block. By applying the unitary channel that permutes the first block according to $\pi\in S(2^n)$, the following matrix is also optimal:
\begin{align*}
& \begin{bmatrix}
    D_1^{\pi} &  \cdots & 0 \\
    \vdots & \ddots & \vdots \\
    0 & \cdots & D_{2^n}
   \end{bmatrix},\\
&\pi\in S(2^n),\qquad
 D_1^{\pi}=\begin{bmatrix}
    \lambda_{\pi(1)} &  \cdots & 0 \\
    \vdots & \ddots & \vdots \\
    0 & \cdots & \lambda_{\pi(2^n)}
   \end{bmatrix}.
\end{align*}
Therefore, using the fact that the set of optimal channels is convex, the following convex combination is also an optimal matrix:
\begin{align*}
 \frac{1}{(2^n)!}\sum_{\pi\in S(2^n)} \begin{bmatrix}
    D_1^{\pi} &  \cdots & 0 \\
    \vdots & \ddots & \vdots \\
    0 & \cdots & D_{2^n}
   \end{bmatrix} &= \begin{bmatrix}
    \beta_1 I &  \cdots & 0 \\
    \vdots & \ddots & \vdots \\
    0 & \cdots & D_{2^n}
   \end{bmatrix},\\
   \beta_1&=\frac{1}{2^n}\sum_{i=1}^{2^n}\lambda_i.
\end{align*}
By repeating the same process for each block,
\begin{align*}
\begin{bmatrix}
    \beta_1 I &  \cdots & 0 \\
    \vdots & \ddots & \vdots \\
    0 & \cdots & \beta_{2^n} I
   \end{bmatrix} = \begin{bmatrix}
    \beta_1&  \cdots & 0 \\
    \vdots & \ddots & \vdots \\
    0 & \cdots & \beta_{2^n}
   \end{bmatrix} \otimes I
\end{align*}
is also optimal. Moreover, since the infimum is taken over matrices of trace 0, we have $\sum_{i=1}^{2^n} \beta_i=0$. We now proceed similarly, but permute the blocks themselves rather than the terms within one block. For $\pi\in S(2^n)$, let $U_{\pi}$ be the corresponding permutation matrix. We consider the quantum channel $\mathcal{C}_{\pi}\coloneqq U_{\pi}(\cdot)U_{\pi}^{\dag} \otimes \mathcal{I}_E$, which preserves the spaces $E_0$ and $E_1$. Since the trace norm is stable under application of unitary channels,
\[
\mathcal{C}_{\pi}\left(\begin{bmatrix}
    \beta_1 &  \cdots & 0 \\
    \vdots & \ddots & \vdots \\
    0 & \cdots & \beta_{2^n}
   \end{bmatrix} \otimes I\right) =
   \begin{bmatrix}
    \beta_{\pi(1)} &  \cdots & 0 \\
    \vdots & \ddots & \vdots \\
    0 & \cdots & \beta_{\pi(2^n)}
   \end{bmatrix} \otimes I
\]
is also optimal. By convexity of the set of optimal channels,
\[
 \frac{1}{(2^n)!}\sum_{\pi\in S(2^n)} \begin{bmatrix}
    \beta_{\pi(1)} &  \cdots & 0 \\
    \vdots & \ddots & \vdots \\
    0 & \cdots & \beta_{\pi(2^n)}
   \end{bmatrix} \otimes I = 0,
\]
is also optimal, which concludes the proof.
\end{proof}
\end{appendices}
\end{document}